\documentclass[10pt,journal, a4paper]{IEEEtran}
\usepackage[utf8]{inputenc}
\usepackage{amsfonts,amsmath,amssymb,amsthm,mathrsfs,bm,bbm}
\usepackage{enumerate}
\usepackage{cite}
\usepackage{graphicx}
\usepackage[caption=false,font=footnotesize]{subfig}
\usepackage{algpseudocode}
\usepackage{algorithm}

\usepackage{blindtext}
\usepackage{etoolbox}

\newtheorem{lemma}{Lemma}

\newtheorem{proposition}{Proposition}

\newtheorem{remark}{Remark}

\usepackage{hyperref}
\hypersetup{
  colorlinks   = true,    % Colours links instead of ugly boxes
  urlcolor     = black,    % Colour for external hyperlinks
  linkcolor    = black,    % Colour of internal links
  citecolor    = red      % Colour of citations
}

\newcommand{\eqdef}{:=}

\newcommand{\rvec}[1]{\mathbbm{#1}} 		% random vectors 
\newcommand{\rmat}[1]{\mathbbm{#1}} 	% random matrices
\newcommand{\E}{\mathsf{E}}		% expectation
\newcommand{\V}{\mathsf{V}}			% variance
\newcommand{\stdset}[1]{\mathbbmss{#1}}	% standard sets
\newcommand{\set}[1]{\mathcal{#1}}		% sets
\renewcommand{\vec}[1]{\bm{#1}}		% deterministic vector or matrix
\newcommand{\CN}{\mathcal{CN}}			% complex Gaussian
\newcommand{\herm}{\mathsf{H}}			% Hermitian transpose
\newcommand{\T}{\mathsf{T}}				% Transpose

\author{Lorenzo Miretti,~\IEEEmembership{Member,~IEEE}, Renato L.~G. Cavalcante,~\IEEEmembership{Member,~IEEE}, Emil Björnson,~\IEEEmembership{Fellow,~IEEE}, Sławomir~Sta\'nczak,~\IEEEmembership{Senior~Member,~IEEE}}
\title{UL-DL Duality for Cell-free Massive MIMO with Per-AP Power and Information Constraints% <-this % stops a space
\thanks{L.~Miretti and S.~Sta\'nczak are with the Network Information Theory group, Technische Universität  Berlin, Berlin 10587, Germany, and the Department of 
Wireless Communications and Networks, Fraunhofer Institute for Telecommunications Heinrich-Hertz-Institut HHI, Berlin 10587, Germany (email: \{miretti, slawomir.stanczak\}@tu-berlin.de). R.~L.~G.~Cavalcante is with the Department of 
Wireless Communications and Networks, Fraunhofer Institute for Telecommunications Heinrich-Hertz-Institut HHI, Berlin 10587, Germany (email: renato.cavalcante@hhi.fraunhofer.de). E.~Bj\"ornson is with the Department of Computer Science, KTH Royal Institute of Technology, SE-100~44 Stockholm, Sweden (email: emilbjo@kth.de). \\ L.~Miretti, R.~L.~G.~Cavalcante, and S.~Sta\'nczak acknowledge the financial support by the Federal Ministry of Education and Research of Germany in the programme of “Souverän. Digital. Vernetzt.” Joint project 6G-RIC, project identification numbers: 16KISK020K, 16KISK030. E.~Bj\"ornson acknowledges the financial support by the FFL18-0277 grant from the Swedish Foundation for Strategic Research.
}\thanks{© 2024 IEEE.  Personal use of this material is permitted.  Permission from IEEE must be obtained for all other uses, in any current or future media, including reprinting/republishing this material for advertising or promotional purposes, creating new collective works, for resale or redistribution to servers or lists, or reuse of any copyrighted component of this work in other works.}}

\begin{document}
\maketitle

%\IEEEpubid{\begin{minipage}{\textwidth}\ \\[12pt] \centering
%  © 2024 IEEE.  Personal use of this material is permitted.  Permission from IEEE must be obtained for all other uses, in any current or future media, including reprinting/republishing this material for advertising or promotional purposes, creating new collective works, for resale or redistribution to servers or lists, or reuse of any copyrighted component of this work in other works.
%\end{minipage}}

\begin{abstract}
We derive a novel uplink-downlink duality principle for optimal joint precoding design under per-transmitter power and information constraints in fading channels. The information constraints model limited sharing of channel state information and data bearing signals across the transmitters. The main application is to cell-free networks, where each access point (AP) must typically satisfy an individual power constraint and form its transmit signal using limited cooperation capabilities. Our duality principle applies to ergodic achievable rates given by the popular hardening bound, and it can be interpreted as a nontrivial generalization of a previous result by Yu and Lan for deterministic channels. This generalization allows us to study involved information constraints going beyond the simple case of cluster-wise centralized precoding covered by previous techniques. Specifically, we show that the optimal joint precoders are, in general, given by an extension of the recently developed team minimum mean-square error method. As a particular yet practical example, we then solve the problem of optimal local precoding design in user-centric cell-free massive MIMO networks subject to per-AP power constraints.
\end{abstract}

\section{Introduction}
\IEEEpubidadjcol
Cell-free massive MIMO networks have attracted significant interest for their potential in enhancing the performance of future generation mobile access networks. The main focus is the evolution of known coordinated multi-point concepts (CoMP) towards practically attractive access solutions that combine the benefits of access point (AP) cooperation and ultra-dense deployments. To this end, considerable research effort has been devoted to the development of scalable and possibly user-centric system architectures and algorithms covering, for instance, power control, pilot-based channel estimation, joint processing such as precoding and combining, fronthaul overhead, network topology, and initial access \cite{ngo2017cellfree,nayebi2017precoding,interdonato2019scalability,bashar2019maxmin,bashar2019quantization,interdonato2019ubiquitous,buzzi2020,emil2020scalable,attarifar2020subset,liu2020analysis,italo2021ota,interdonato2021enhanced,dandrea2021soft,lancho2021short,du2021cellfree,shaik2021mmse,chen2021access,gottsch2022subspace}.      

Against this background, in this study we address the open problem of optimal joint downlink precoding design in cell-free massive MIMO networks, by considering minimum quality-of-service requirements and by assuming that each AP is subject to an individual power constraint and an individual information constraint. The information constraints model limited AP cooperation capabilities, which are motivated by the need for realizing scalable cell-free architectures with reduced fronthaul and joint processing load. More specifically, the information constraints model:
\begin{itemize}
\item limited sharing of channel state information (CSI), covering, for instance, the local CSI model as in the original version of the cell-free massive MIMO paradigm \cite{ngo2017cellfree}; and
\item limited sharing of data bearing signals, for instance, only within user-centric cooperation clusters as in \cite{buzzi2020,emil2020scalable}.
\end{itemize}
Each AP must form its transmit signal as a function of the CSI and data bearing signals specified by the constraints, and no additional information exchange between the APs is allowed. This is in contrast to related works such as \cite{italo2021ota}, which covers iterative information exchange during precoding computation. In the cell-free massive MIMO literature, the best performing joint precoders are typically designed from joint uplink combiners, motivated by a known uplink-downlink duality principle for fading channels \cite[Ch.~4]{massivemimobook} \cite[Ch.~6]{demir2021foundations}. However, optimal joint precoders are generally unknown owing to the following two reasons. First, until very recently, optimal joint combiners were not known except for the relatively simple case of full CSI sharing within each cooperation cluster, an information constraint leading to so-called \textit{centralized} combining. Second, the known uplink-downlink duality principle for fading channels holds for a looser and somewhat less practical sum power constraint. 

\IEEEpubidadjcol
Addressing the first issue is known to be quite challenging. In essence, depending on the information constraints, each AP may need to take combining decisions that are robust not only to channel estimation errors, but also to the possibly unknown combining decisions taken at the other APs. A novel method for optimally addressing this issue, called the \textit{team} minimum mean-square error (MMSE) method, has been recently proposed in \cite{miretti2021team}. Together with the available uplink-downlink duality principle for fading channels, this method is used in \cite{miretti2021team,miretti2021team2} to derive novel \textit{distributed} precoders under various relevant information constraints, including the aforementioned local CSI model. However, in its current form, this method is provenly optimal under a sum power constraint only, and hence it does not provide a solution to our problem.

%, and for  variations based on equivalent expressions of the optimal centralized solution \cite{shaik2021mmse,italo2021ota}. However, these variations require either additional radio resources for iterative bidirectional over-the-air training procedure, while \cite{shaik2021mmse} exploits properties of sequential information processing in serial fronthauls that cannot be transposed to downlink precoding.
A partial solution to the second issue is given by the alternative uplink-downlink duality principle under per-antenna power constraints developed in \cite{yu2007transmitter}, applied for instance in \cite{emil2013optimal} in the context of CoMP. However, the method in \cite{yu2007transmitter} applies to deterministic channels, i.e., to fixed channel fading realizations, which imposes many limitations. A first limitation is that optimal schemes typically involve solving relatively complex optimization problems for each channel realization, which may be impractical in large systems. Perhaps the most important limitation is that the technique in \cite{yu2007transmitter} applies to centralized precoding only, i.e.,  it does not cover optimal distributed precoding design  under limited CSI sharing. This crucial point can be informally explained as follows.

Consider the following classical feasibility problem over a $L\times K$ MISO broadcast channel with per-antenna power constraints:
\begin{equation}\label{eq:deterministic}
\begin{aligned}
\text{find} \quad & \vec{T} \in \stdset{C}^{L\times K}\\
\text{subject to} \quad & (\forall k \in \{1,\ldots,K\})~r_k(\vec{H},\vec{T}) \geq R,\\
& \textstyle(\forall l \in \{1,\ldots,L\})~\sum_{k=1}^K|[\vec{T}]_{l,k}|^2\leq P,
\end{aligned}
\end{equation}
where $r_k(\vec{H},\vec{T})=\log(1+\mathrm{SINR}_k(\vec{H},\vec{T}))$ denotes the $k$th user's instantaneous achievable rate for a fixed channel matrix $\vec{H}$ and precoding matrix $\vec{T}$. A solution to this problem is given in \cite{yu2007transmitter}, and it takes the form of a properly regularized pseudo-inverse of $\vec{H}$, i.e., a standard MMSE precoding matrix. Now, in order to enforce a nontrivial information constraint on $\vec{T}$ going beyond the case of cluster-wise centralized precoding, such as letting each entry of $\vec{T}$ to depend on  different channel estimates, the above problem needs to be radically modified. One possible approach is to formulate the problem in a statistical sense, by considering ergodic rate optimization over fading channels\footnote{Other formulations may include the replacement of the average constraints with almost sure inequalities, or deterministic worst-case approaches, i.e., to optimize each entry of $\vec{T}$ such that the  rate and power constraints are satisfied for all possible values of the other entries of $\vec{T}$. However, these formulations are typically  intractable and may be too conservative. Hence, they are not covered by both the current literature and this study.}:
\begin{equation}\label{eq:ergodic}
\begin{aligned}
\text{find} \quad & \rvec{T} \in \set{T}\\
\text{subject to} \quad & (\forall k \in \{1,\ldots,K\})~\E[r_k(\rvec{H},\rvec{T})] \geq R,\\
& \textstyle(\forall l \in \{1,\ldots,L\})~\sum_{k=1}^K\E\left[|[\rvec{T}]_{l,k}|^2\right]\leq P,
\end{aligned}
\end{equation}
where $\rvec{T}$ and $\rvec{H}$ are random matrices, and where the information constraint is encoded in $\set{T}$ as a certain subset of the space of $L\times K$ random matrices (see Section~\ref{sec:model} for details). However, the technique in \cite{yu2007transmitter} does not cover fading channels, i.e., it cannot solve \eqref{eq:ergodic}.

%This is because standard optimization problems taking as input a fixed channel realization produce solutions that may depend on fixed yet unknown channel variables, and hence violate the information constraints. 

To address the above issues, in this study\footnote{A part of the results in this study is presented in \cite{miretti2022duality} without proof. This study extends \cite{miretti2022duality} by providing complete derivations, the expressions for optimal local and centralized precoding, and additional details on the numerical implementation of the proposed algorithms.} we derive a novel uplink-downlink duality principle for fading channels under per-AP power constraints. More precisely, we extend the technique in \cite{yu2007transmitter} to cover a variation of \eqref{eq:ergodic} with $\E[r_k(\rvec{H},\rvec{T})]$ replaced by a more tractable lower bound where the expectation is rigorously moved inside the logarithm following  the so-called \textit{hardening} bounding technique \cite{marzetta2016fundamentals}. As discussed in details throughout this study, our derivation significantly departs from \cite{yu2007transmitter}, mostly due to the challenges introduced while moving from classical optimization problems over the (finite dimensional) space of deterministic matrices, as in \eqref{eq:deterministic}, to more involved functional optimization problems over the (infinite dimensional) space of random matrices, as in \eqref{eq:ergodic}. Furthermore, building on the above result, we show that optimal joint precoders are given by solutions to properly parametrized MMSE problems under per-AP information constraints, i.e., by properly parametrized variations of team MMSE precoders \cite{miretti2021team}. In summary, the main contribution of this study can be interpreted as a nontrivial extension of the method in \cite{yu2007transmitter} to fading channels, and of the method in \cite{miretti2021team} to per-AP power constraints. As a concrete application of our findings, we then solve the previously open problem of optimal \textit{local} precoding design in cell-free massive MIMO networks under per-AP power constraints. Moreover, we provide a potentially simpler variation of the known centralized solution under per-AP power constraints where, in contrast to previous techniques, its parameters are not optimized for each channel realization, but based on relatively slowly varying channel statistics. 

The paper is organized as follows. Section~\ref{sec:math} and Section~\ref{sec:model} provide the main definitions, mathematical tools, and modeling assumptions. Section~\ref{sec:lagrange} presents and studies the main optimization problem using Lagrangian duality arguments. Building on the obtained insights, Section~\ref{sec:duality} derives the proposed uplink-downlink duality principle, which is then exploited in Section~\ref{sec:applications} to characterize the optimal solution structure. Simple applications of the main results are illustrated in Section~\ref{sec:sim} by means of numerical simulations. Finally, Section~\ref{sec:conclusion} summarizes the main results, and outlines some limitations and possible future directions. 

\textit{Reproducible Research}: the simulation code is available at \url{https://github.com/lorenzomiretti/duality}.

\section{Mathematical preliminaries}
\label{sec:math}
\subsection{Notation and definitions}
We denote by $\stdset{R}_+$ and $\stdset{R}_{++}$ the sets of, respectively, nonnegative and positive reals. The Euclidean norm in $\stdset{C}^K$ is denoted by $\|\cdot\|$. Let $(\Omega,\Sigma,\mathbb{P})$ be a probability space. We denote by $\set{H}^K$ the set of complex-valued random vectors, i.e., $K$-tuples of $\Sigma$-measurable functions $\Omega \to \stdset{C}$ satisfying $(\forall \rvec{x}\in \set{H}^K)$ $\E[\|\rvec{x}\|^2]<\infty$. Together with the standard operations of addition and real scalar multiplication, we recall that $\set{H}^K$ is a real vector space. Given a random variable $X\in \set{H}$, we denote by $\E[X]$ and $\V(X)$ its expected value and variance, respectively. We use $\rvec{x}-\rvec{y}-\rvec{z}$ to denote a Markov chain, i.e., to denote that the random vectors $\rvec{x}$ and $\rvec{z}$ are conditionally independent given another random vector $\rvec{y}$. Inequalities involving vectors in $\stdset{R}^K$ should be understood coordinate-wise. The $k$th column of the $K$-dimensional identity matrix $\vec{I}_K$ is denoted by $\vec{e}_k$. 

\subsection{Lagrangian duality in general vector spaces} 
The following key result is frequently invoked throughout our study, and can be found in \cite{zalinescu2002convex}.
\begin{proposition}\label{prop:duality}
Consider the functions $f: \set{X} \to \stdset{R}$ and $\vec{g}: \set{X} \to \stdset{R}^N$, where $\set{X}$ is a real vector space, and the optimization problem
\begin{equation*}
\begin{aligned}
\underset{X \in \set{X}} {\text{minimize}} \quad &  f(X)\\
\text{subject to} \quad & \vec{g}(X)\leq \vec{0}.
\end{aligned}
\end{equation*}
Define the primal optimum $p^\star := \inf\{f(X)~|~\vec{g}(X)\leq \vec{0}, X\in\set{X}\}$, and the dual optimum $ d^\star := \sup\{d(\vec{\lambda})~|~\vec{\lambda}\in \stdset{R}^N_+\}$, where $d(\vec{\lambda}):=\inf\{f(X)+\vec{\lambda}^\T\vec{g}(X)~|~ X\in\set{X}\}$ for $\vec{\lambda}\in \stdset{R}_+^N$. Each of the following holds: 
\begin{enumerate}[(i)]
\item $d^\star \leq p^\star$ (weak duality);
\item If $f$ and $\vec{g}$ are proper convex functions \cite[pp. 39]{zalinescu2002convex}, and $\{X\in \set{X} ~|~ \vec{g}(X) < \vec{0}\}\neq \emptyset$ (Slater's condition), then $d^\star = p^\star$ holds (strong duality). Furthermore, there exist Lagrangian multipliers $\vec{\lambda}^\star \in \stdset{R}_+^N$ such that $d(\vec{\lambda}^\star) = d^\star$.
\end{enumerate} 
\end{proposition} 
\begin{proof}
For statement (i), see \cite[Theorem 2.6.1(iii)]{zalinescu2002convex}. For statement (ii), see \cite[Theorem 2.9.3(ii)]{zalinescu2002convex}.
\end{proof}
For $\set{X} = \stdset{R}^K$, the above proposition corresponds to the familiar Lagrangian duality principle for optimization problems in Euclidean spaces \cite{boyd2004convex}. In this study, we exploit the more general result in Proposition~\ref{prop:duality} to address optimization problems in the real vector space of complex-valued random vectors $\set{X}=\set{H}^K$. 

\section{System model}
\label{sec:model}
\subsection{Downlink achievable rates}
Consider the downlink of a cell-free wireless network composed of $L$ APs indexed by $\mathcal{L}:=\{1,\ldots,L\}$, each of them equipped with $N$ antennas, and $K$ single-antenna UEs indexed by $\mathcal{K}:=\{1,\ldots,K\}$. By assuming  a standard synchronous and frequency-flat channel model governed by an ergodic and stationary fading process, and simple transmission techniques based on 
linear precoding and on treating interference as noise, we focus on simultaneously achievable ergodic rates in the classical Shannon sense, approximated by the popular \emph{hardening} inner bound \cite{marzetta2016fundamentals}. In more detail, we define the downlink rates achieved by each UE $k\in \set{K}$ for a given precoding design as  
\begin{equation}\label{eq:hardening_bound}
R_k^{\mathrm{DL}}(\rvec{T}) \eqdef  \log\left(1+\mathrm{SINR}_k^{\mathrm{DL}}(\rvec{T})\right),
\end{equation}
\begin{equation}\mathrm{SINR}_k^{\mathrm{DL}}(\rvec{T}) \eqdef \dfrac{|\E[\rvec{h}_k^\herm\rvec{t}_k]|^2}{\V(\rvec{h}_k^\herm\rvec{t}_k) + \sum_{j\neq k}\E[|\rvec{h}_k^\herm\rvec{t}_j|^2]+1},
\end{equation}
where $\rvec{h}_k \in\set{H}^{NL}$ is a random channel vector modeling the fading state between UE $k$ and all APs, $\rvec{t}_k \in \set{H}^{NL}$ is a joint precoding vector applied by all APs to the coded and modulated data bearing signal for UE~$k$, and $\rvec{T} \eqdef [\rvec{t}_1,\ldots,\rvec{t}_K]\in \set{H}^{NL\times K}$ is the aggregate joint precoding matrix. We stress that precoders are defined and denoted as random quantities, since they may adapt to random fading realizations on the basis of the the available instantaneous CSI. This aspect is treated in detail in the next sections.  

\subsection{Per-AP power and information constraints}
\label{ssec:info}
In practical cell-free wireless networks, each AP must typically satisfy an individual power constraint. In addition, motivated by the need for realizing scalable cell-free architectures, each AP is also typically subject to an individual information constraint induced by limited data and instantaneous CSI sharing, which impair its cooperation capabilities. In this work, the above per-AP constraints are modelled as follows. Let $[\rvec{t}_{1,k}^\T,\ldots,\rvec{t}_{L,k}^\T]^\T \eqdef \rvec{t}_k$, where $\rvec{t}_{l,k} \in \set{H}^N$ denotes the portion of the precoder applied by AP~$l\in\set{L}$ to serve UE~$k\in\set{K}$. By assuming unitary power data bearing signals, we consider the average power constraints
\begin{equation}
(\forall l \in \set{L})~\sum_{k=1}^K\E[\|\rvec{t}_{l,k}\|^2]\leq P_l < \infty.
\end{equation}
We choose average power constraints instead of perhaps more common instantaneous power constraints mostly for tractability reasons. This is in line with the related results on uplink-downlink duality for fading channels under a sum power constraint \cite{massivemimobook}. However, we point out that average power constraints may be in fact quite appropriate for modern wideband systems, where power allocation over multiple subcarriers is standard practice.

For modeling impairments related to limited CSI sharing, we follow the recently proposed approach in \cite{miretti2021team} and let \begin{equation}\label{eq:CSI}
(\forall k \in \set{K})~\rvec{t}_k\in \set{T}_k \eqdef \set{H}_1^N \times \ldots \times \set{H}_L^N,
\end{equation}
where $\set{H}_l^N\subseteq \set{H}^N$ denotes the set of $N$-tuples of $\Sigma_l$-measurable functions $\Omega \to \stdset{C}$ satisfying $(\forall \rvec{x}\in \set{H}_l^N)$ $\E[\|\rvec{x}\|^2]<\infty$, and where $\Sigma_l \subseteq \Sigma$ is the sub-$\sigma$-algebra induced by the available CSI at AP $l\in \set{L}$, also called the \emph{information subfield} of AP $l$ \cite{yukselbook}. Informally, by letting $S_l$ be a given random variable modeling the available CSI at AP $l\in \set{L}$, this constraint enforces the precoders $\{\rvec{t}_{l,k}\}_{k\in \set{K}}$ of the $l$th AP to be functions of $S_l$ only. Limited CSI sharing typically leads to the case $S_l \neq S_j$ (and hence $\Sigma_l\neq \Sigma_j$) for some $j\neq l$.  
\begin{remark}
\label{rem:CSI}
The constraint in \eqref{eq:CSI} is fairly general. For example, it covers the case of local CSI \cite{ngo2017cellfree} (i.e., where each AP has information on only the channel between the UEs and itself), but also more advanced cooperation structures exploiting the peculiarities of efficient fronthauls such as in the so-called radio stripes concept where the APs are daisy-chained \cite{miretti2021team}. More precisely, by letting $\hat{\rvec{H}}_l$ denote the local estimate of the local channel matrix from AP $l\in \set{L}$ to all UEs, the following cases are modeled and studied in detail using \eqref{eq:CSI} in \cite{miretti2021team}:
\begin{itemize}
\item (local CSI) $(\forall l\in \set{L})~S_l = \hat{\rvec{H}}_l$;
\item (unidirectional CSI) $(\forall l\in \set{L})~S_l = (\hat{\rvec{H}}_1,\ldots,\hat{\rvec{H}}_l)$;
\item (centralized CSI) $(\forall l\in \set{L})~S_l = (\hat{\rvec{H}}_1,\ldots,\hat{\rvec{H}}_L)$. 
\end{itemize}
In addition, as mentioned in \cite{miretti2021team}, \eqref{eq:CSI} may also cover more general cases where, owing to fronthaul imperfections such as quantization or delay, the APs can only obtain degraded versions of the local CSI estimates from  other APs. To keep our results general, in most of this study we do not specify the CSI structure. However, Section~\ref{sec:applications} discusses concrete applications of \eqref{eq:CSI}. More precisely, it studies the cases of local CSI and centralized CSI in detail.
\end{remark}
From a mathematical point of view, $\set{H}_l^N$ is a subspace of the real vector space $\set{H}^N$, which in turn implies that the constraint in $\eqref{eq:CSI}$ models limited CSI sharing by constraining $\rvec{t}_k$ within a subspace $\set{T}_k$ of the real vector space $\set{H}^{LN}$ \cite{yukselbook}. Furthermore, following the well-known user-centric network clustering approach \cite{buzzi2020,emil2020scalable}, we assume that each UE $k \in \set{K}$ is only served by a subset $\set{L}_k\subseteq \set{L}$ of APs. As shown in \cite{miretti2021team2}, this additional practical impairment can be straightforwardly included in \eqref{eq:CSI} by replacing $\set{H}_l^N$ with the set $\{\vec{0}_N~ (\mathrm{a.s.})\}$ for each AP~$l \notin \set{L}_k$, i.e., not serving UE~$k$. Since $\{\vec{0}_N~ (\mathrm{a.s.})\}$ is a (trivial) subspace of $\set{H}^N$, this replacement does not alter the property of $\set{T}_k$ being a real vector space (in fact, a subspace of $\set{H}^{NL}$). This key property will allow us to address the problem of optimal joint precoding design under per-AP information constraints using Lagrangian duality for real vector spaces (Proposition~\ref{prop:duality}). 

\subsection{Dual uplink rates under arbitrary noise powers}\label{ssec:dual_UL_model}
The main optimization approach developed in this work has a natural interpretation in terms of a virtual dual uplink channel with arbitrary noise power. More specifically, similarly to the chosen downlink model, we consider virtual uplink ergodic rates given by the \textit{use-and-then-forget} (UatF) inner bound \cite{marzetta2016fundamentals}  
\begin{equation}\label{eq:UatF_bound}
R_k^{\mathrm{UL}}(\rvec{v}_k,\vec{p},\vec{\sigma}) \eqdef  \log\left(1+\mathrm{SINR}_k^{\mathrm{UL}}(\rvec{v}_k,\vec{p},\vec{\sigma}) \right),
\end{equation}
\begin{equation}
\begin{split}&\mathrm{SINR}_k^{\mathrm{UL}}(\rvec{v}_k,\vec{p},\vec{\sigma}) \eqdef \\ 
&\dfrac{p_k|\E[\rvec{h}_k^\herm\rvec{v}_k]|^2}{p_k\V(\rvec{h}_k^\herm\rvec{v}_k) + \sum_{j\neq k}p_j\E[|\rvec{h}_j^\herm\rvec{v}_k|^2]+\E[\|\rvec{v}_k\|_{\vec{\sigma}}^2]},
\end{split}
\end{equation}
where $\rvec{v}_k=[\rvec{v}_{1,k}^\T,\ldots,\rvec{v}_{L,k}^\T]^\T\in \set{H}^{NL}\backslash\{\vec{0}\}$ is a joint combiner, $\vec{p} \eqdef [p_1,\ldots,p_K]^\T \in \stdset{R}_{+}^K$ is a vector of transmit powers, and where we define $(\forall \rvec{v}_k \in \set{H}^{NL})(\forall \vec{\sigma} \in \stdset{R}_{++}^L)$
\begin{equation}\label{eq:norm}
\E[\|\rvec{v}_k\|_{\vec{\sigma}}^2] := \sum_{l=1}^L\sigma_l\E[\|\rvec{v}_{l,k}\|^2]
\end{equation}
for given $\vec{\sigma}\eqdef[\sigma_1,\ldots,\sigma_L]^\T \in \stdset{R}_{++}^L$. In the above expressions, $\vec{\sigma}$ can be interpreted as a vector collecting uplink noise powers for each AP. We remark that the term \textit{virtual} here refers to the fact that the above rates may not be achievable in the true uplink channel, since $(\vec{p},\vec{\sigma})$ may differ from the true uplink transmit and noise powers. The major difference between $R_k^{\mathrm{UL}}$ and $R_k^{\mathrm{DL}}$ is that the former depends only on the joint combiner $\rvec{v}_k$ for the signal of UE~$k\in\set{K}$, while the latter depends on the entire precoding matrix $\rvec{T}$. In fact, the uplink achievable rates are coupled only via the transmit and noise powers $\vec{p},\vec{\sigma}$. This known aspect makes optimization on the uplink channel generally easier than on the downlink channel. 

\section{Problem statement and Lagrangian duality}
\label{sec:lagrange}
To address the problem of optimal joint precoding design in cell-free networks, in this section we study a certain optimization problem subject to minimum SINR requirements and per-AP power and information constraints. In particular, given a tuple of power constraints $(P_1,\ldots,P_L)\in \stdset{R}_{++}^K$ and of SINR requirements $(\gamma_1,\ldots,\gamma_K)\in \stdset{R}_{++}^K$, we consider the following infinite dimensional optimization problem:
\begin{equation}\label{prob:DL_QoS}
\begin{aligned}
\underset{\rvec{T} \in \set{T}} {\text{minimize}} \quad &  \sum_{k=1}^K\E[\|\rvec{t}_k\|^2]\\
\text{subject to} \quad & (\forall k \in \set{K})~\mathrm{SINR}^{\mathrm{DL}}_k(\rvec{T}) \geq \gamma_k,\\
& (\forall l \in \set{L})~\sum_{k=1}^K\E[\|\rvec{t}_{l,k}\|^2]\leq P_l,
\end{aligned}
\end{equation}
where $\set{T}\subset \set{H}^{NL\times K}$ is a real vector space obtained by collecting all per-AP information constraints defined in Section~\ref{ssec:info}. We recall that these constraints accommodate both limited instantaneous CSI sharing and user-centric network clustering. 
In the following, to avoid technical digressions, we assume that strictly feasible joint precoders exist, i.e., 
\begin{equation*}
\left\{ \rvec{T}\in\set{T} \middle| \begin{array}{l}
(\forall k \in \set{K}) ~\mathrm{SINR}^{\mathrm{DL}}_k(\rvec{T}) > \gamma_k \\
(\forall l \in \set{L})~\sum_{k=1}^K\E[\|\rvec{t}_{l,k}\|^2]< P_l
\end{array} \right\} \neq \emptyset.
\end{equation*} 
Furthermore, due to both mathematical convenience and practical reasons, in Problem~\eqref{prob:DL_QoS} we focus on the subset of feasible joint precoders minimizing the total power consumption.

\subsection{Lagrangian dual problems}
Inspired by the related results in \cite{yu2007transmitter} based on Lagrangian duality for finite dimensional optimization problems, in this section we apply Lagrangian duality for infinite dimensional optimization problems to study  Problem~\eqref{prob:DL_QoS}. For convenience, we adopt the compact notation 
\begin{equation}
\mathrm{SINR}_k^{\mathrm{DL}}(\rvec{T}) = \dfrac{|b_k(\rvec{T})|^2}{|c_k(\rvec{T})|^2-|b_k(\rvec{T})|^2},
\end{equation}
where $b_k(\rvec{T}) \eqdef \E[\rvec{h}_k^\herm\rvec{t}_k]$ is the useful signal term, and $|c_k(\rvec{T})|^2 \eqdef \sum_{j=1}^K\E[|\rvec{h}_k^\herm\rvec{t}_j|^2]+1$ is the interference plus noise power term. Furthermore, we rearrange the SINR constraints using simple algebraic manipulations, leading to the following simple property $(\forall \rvec{T}\in\set{T})(\forall k\in\set{K})$ 
\begin{equation}\label{eq:SINR_rearranged}
\mathrm{SINR}_k^{\mathrm{DL}}(\rvec{T})\geq \gamma_k \iff |c_k(\rvec{T})|-\nu_k|b_k(\rvec{T})|\leq 0,
\end{equation}
where $\nu_k\eqdef \sqrt{1+1/\gamma_k}$. 
A Lagrangian dual problem to \eqref{prob:DL_QoS} is then given by
\begin{equation}\label{prob:DL_QoS_dual}
\underset{(\vec{\lambda},\vec{\mu})\in \stdset{R}_+^L\times \stdset{R}_+^K} {\text{maximize}}~d(\vec{\lambda},\vec{\mu}),
\end{equation}
where (recalling \eqref{eq:norm}) we define the dual function $d(\vec{\lambda},\vec{\mu}) \eqdef$
\begin{equation*}
\inf_{\rvec{T} \in \set{T}}  \sum_{k=1}^K \E[\|\rvec{t}_k\|_{\vec{1}+\vec{\lambda}}^2] -\sum_{l=1}^L \lambda_l P_l + \sum_{k=1}^K\mu_k(|c_k(\rvec{T})|-\nu_k|b_k(\rvec{T})|).
\end{equation*}
Since the primal problem in \eqref{prob:DL_QoS} is nonconvex, by Proposition~\ref{prop:duality} we can only guarantee (a-priori) weak duality. Furthermore, guaranteeing existence of a solution is not immediate. However, the following important result holds.
\begin{proposition}\label{prop:strong_duality}
Problem~\eqref{prob:DL_QoS} admits a solution $\rvec{T}^\star\in \set{T}$. Furthermore, denote by $p^\star$ and $d^\star$ the optimum of the primal problem \eqref{prob:DL_QoS} and of the dual problem \eqref{prob:DL_QoS_dual}, respectively. Strong duality holds, i.e., $d^\star = p^\star$, and there exist Lagrangian multipliers $(\vec{\lambda}^\star,\vec{\mu}^\star)\in \stdset{R}_+^L\times \stdset{R}_+^K$ solving Problem~\eqref{prob:DL_QoS_dual}.
\end{proposition}
\begin{proof}
The proof follows a similar idea as in \cite{yu2007transmitter,wiesel2006linear}, and it is based on establishing connections between Problem~\eqref{prob:DL_QoS} and a convex reformulation obtained by replacing $|b_k(\rvec{T})|$ in \eqref{eq:SINR_rearranged} with $\Re(b_k(\rvec{T}))$. For additional details, see Appendix~\ref{app:convex_reformulation}. 
\end{proof}

We conclude this section by stating a useful consequence of Proposition~\ref{prop:strong_duality} that will be instrumental for proving our main results based on uplink-downlink duality. In particular, we provide an alternative version of Proposition~\ref{prop:strong_duality} based on a \emph{partial} dual problem, obtained by keeping the SINR constraints implicit.
\begin{proposition}\label{lem:partial_dual}
Given the subset $\set{T}_{\vec{\gamma}} \eqdef \{\rvec{T}\in\set{T} ~|~ (\forall k \in \set{K})~\mathrm{SINR}_k^{\mathrm{DL}}(\rvec{T}) \geq \gamma_k\}$ of precoders satisfying the SINR constraints in Problem~\eqref{prob:DL_QoS}, define the \emph{partial} dual problem 
\begin{equation}\label{prob:partial_dual}
\underset{\vec{\lambda}\in \stdset{R}_+^L} {\text{maximize}}~\tilde{d}(\vec{\lambda}) \eqdef \inf_{\rvec{T} \in \set{T}_{\vec{\gamma}}} \sum_{k=1}^K \E[\|\rvec{t}_k\|_{\vec{1}+\vec{\lambda}}^2] -\sum_{l=1}^L \lambda_l P_l.
\end{equation}
Strong duality holds, i.e., Problem~\eqref{prob:DL_QoS} and Problem~\eqref{prob:partial_dual} have the same optimum $p^\star$. Furthermore, there exist Lagrangian multipliers $\vec{\lambda}^\star$ solving Problem~\eqref{prob:partial_dual}.
\end{proposition}
\begin{proof}
Consider the alternative optimization problem 
\begin{equation}\label{prob:indicator}
\begin{aligned}
\underset{\rvec{T} \in \set{T}} {\text{minimize}} \quad &  \sum_{k=1}^K\E[\|\rvec{t}_{k}\|^2] + \gamma(\rvec{T})\\
\text{subject to} \quad & (\forall l \in \set{L})~\sum_{k=1}^K\E[\|\rvec{t}_{l,k}\|^2]\leq P_l,
\end{aligned}
\end{equation}
where $\gamma(\rvec{T}) = 0$ if $\rvec{T}$ belongs to the set $\set{T}_{\vec{\gamma}}$, and $\gamma(\rvec{T}) = +\infty$ otherwise. Problem~\eqref{prob:indicator} is equivalent to Problem~\eqref{prob:DL_QoS}, in the sense that it has the same optimum $p^\star$ and set of  optimal solutions. Its Lagrangian dual problem can be written as \eqref{prob:partial_dual}. By applying weak duality (see Proposition~\ref{prop:duality}) to the term $\inf_{\rvec{T} \in \set{T}_{\vec{\gamma}}} \sum_{k=1}^K \E[\|\rvec{t}_k\|_{\vec{1}+\vec{\lambda}}^2]$ for any fixed $\vec{\lambda}$, and by rewriting $\set{T}_{\vec{\gamma}}$ according to \eqref{eq:SINR_rearranged}, we obtain
$\sup_{\vec{\mu}\in\stdset{R}_+^K}d(\vec{\lambda},\vec{\mu})\leq \tilde{d}(\vec{\lambda})$. Taking the supremum over $\vec{\lambda}$ on both sides gives
\begin{equation*}
p^\star = \sup_{\vec{\lambda}\in\stdset{R}_+^L}\sup_{\vec{\mu}\in\stdset{R}_+^K}d(\vec{\lambda},\vec{\mu}) \leq \sup_{\vec{\lambda}\in\stdset{R}_+^L} \tilde{d}(\vec{\lambda}) \leq p^\star,
\end{equation*}
where the first equality follows from Proposition~\ref{prop:strong_duality} and the last inequality follows from weak duality applied to Problem~\eqref{prob:indicator}. The second part of the statement follows from the existence of a solution $(\vec{\lambda}^\star,\vec{\mu}^\star)$ to Problem~\eqref{prob:DL_QoS_dual}, and $p^\star = \sup_{\vec{\mu}\in\stdset{R}_+^K}d(\vec{\lambda}^\star,\vec{\mu}) \leq  \tilde{d}(\vec{\lambda}^\star) \leq p^\star$.
\end{proof}
In the reminder of this study, we focus on the partial dual problem~\eqref{prob:partial_dual}. The reason is that the dual problem~\eqref{prob:DL_QoS_dual}, or other  variations that keep the SINR constraints augmented, do not seem tractable, mostly because of the difficulties in solving the (infinite dimensional) inner minimization problem over the space of precoders subject to nontrivial information constraints. These problems are known to be very challenging, even if convex \cite[Chapter~2]{yukselbook}. In contrast, as we will see later in the manuscript, the use of the partial dual problem leads to a tractable inner minimization problem. This is one of the major differences with respect to the finite dimensional case in \cite{yu2007transmitter}.

\subsection{Primal-dual solution methods}
A key aspect of Lagrangian optimization is the possibility of recovering a primal solution from a dual solution. However, we emphasize that this is not always possible even if strong duality holds. Nevertheless, the following proposition ensures that a primal solution $\rvec{T}^\star$ to Problem~\eqref{prob:DL_QoS} can be indeed recovered from a solution to the partial dual problem~\eqref{prob:partial_dual}.
\begin{proposition}\label{prop:primal-dual}
Let $\vec{\lambda}^\star$ be a solution to Problem~\eqref{prob:partial_dual}. Then, a solution to Problem~\eqref{prob:DL_QoS} is given by any solution to
\begin{equation}\label{prob:DL_QoS_augmented}
\underset{\rvec{T}\in\set{T}_{\vec{\gamma}}} {\text{minimize}}~\sum_{k=1}^K\E[\|\rvec{t}_k\|_{\vec{1}+\vec{\lambda}^\star}^2].
\end{equation}
\end{proposition}
\begin{proof}
The proof is given in Appendix~\ref{app:primaldual}, and it is based on the same convex reformulation of the SINR constraints adopted in the proof of Proposition~\ref{prop:strong_duality}.
\end{proof}
Starting from Proposition~\ref{prop:primal-dual}, and in particular by studying and solving Problem~\eqref{prob:DL_QoS_augmented}, in the next section we derive structural properties for optimal joint precoding. However, before moving to the next section, we first complete the discussion on recovering a primal solution from a dual solution by illustrating a simple algorithm for solving Problem~\eqref{prob:partial_dual}, which is a concave maximization problem. In particular, we consider a standard primal-dual iterative algorithm based on the projected subgradient method \cite{polyak1969minimization,nesterov2003introductory}.
\begin{proposition}\label{lem:proj_grad}
Choose $\vec{\lambda}^{(1)} \in \stdset{R}_{+}^L$ and a sequence $\{\alpha^{(i)}\}_{i\in \stdset{N}}$ such that
\begin{equation*}
(\forall i \in \stdset{N})~\alpha^{(i)}\in\stdset{R}_{++}, \quad \lim_{i\to \infty} \alpha^{(i)} = 0, \quad \sum_{i\in \stdset{N}} \alpha^{(i)} = \infty.
\end{equation*}
Define the sequence $(\vec{\lambda}^{(i)})_{i\in\stdset{N}}$ generated via 
\begin{equation*}
(\forall i \in \stdset{N})~\vec{\lambda}^{(i+1)} = \max\left\{\vec{\lambda}^{(i)} + \frac{\alpha^{(i)}}{\|\vec{g}(\vec{\lambda}^{(i)})\|}\vec{g}(\vec{\lambda}^{(i)}) ,\vec{0}\right\},
\end{equation*}
where the $l$th entry of $\vec{g}(\vec{\lambda})$ is given by $(\forall \vec{\lambda}\in \stdset{R}_+^L)(\forall l\in \set{L})$,
\begin{equation*}\label{eq:supergradient}
g_l(\vec{\lambda})\eqdef\sum_{k=1}^K\E[\|\rvec{t}_{l,k}\|^2]- P_l, \quad 
\rvec{T} \in \arg\min_{\set{T}_{\vec{\gamma}}}\sum_{k=1}^K\E[\|\rvec{t}_k\|_{\vec{1}+\vec{\lambda}}^2].
\end{equation*}
Then, the subsequence of $(\vec{\lambda}^{(i)})_{i\in\stdset{N}}$ corresponding to the best objective after $n$ iterations $\max_{i=1,\ldots,n}\tilde{d}(\vec{\lambda}^{(i)})$  converges to a solution $ \vec{\lambda}^\star$ to Problem~\eqref{prob:partial_dual}.
\end{proposition}
\begin{proof} The proof is given in Appendix~\ref{app:proof_subgradient}.
\end{proof}
Note that the above algorithm requires a method for solving Problem~\eqref{prob:DL_QoS_augmented} for arbitrary Lagrangian multipliers. A possible algorithm is provided in the following section.

\section{Uplink-downlink duality}
\label{sec:duality}
Building on the above analysis based on Lagrangian duality, in this section we present our main result, which states that the problem of optimal joint precoding design under per-AP power and information constraints can be reformulated as a joint combining design and long-term power control problem in a dual uplink channel with a properly designed noise vector $\vec{\sigma}$ (see Section \ref{ssec:dual_UL_model}). More precisely, we show later in Proposition~\ref{prop:UL-DL} that an optimal solution to Problem~\eqref{prob:DL_QoS} can be recovered from a solution to 
\begin{equation}\label{prob:dualUL}
\begin{aligned}
\underset{\rvec{V} \in \set{T}, \vec{p} \in \stdset{R}_+^K} {\text{minimize}} \quad  &  \sum_{k=1}^Kp_k\\
\text{subject to} \quad & (\forall k \in \set{K})~\mathrm{SINR}^{\mathrm{UL}}_k(\rvec{v}_k,\vec{p},\vec{\sigma}) \geq \gamma_k,\\
& (\forall k \in \set{K})~\E[\|\rvec{v}_k\|_{\vec{\sigma}}^2] = 1,
\end{aligned}
\end{equation} 
for some $\vec{\sigma} \in \stdset{R}_{++}^L$ (recalling \eqref{eq:norm}), and where $\rvec{V} := [\rvec{v}_1,\ldots,\rvec{v}_k]$.
In addition, we present an efficient numerical method that solves the above problem.

\begin{remark} Our derivation differs significantly from the derivation of the related result in \cite{yu2007transmitter}. In particular, \cite{yu2007transmitter} exploits the peculiar structure of optimal centralized precoding in deterministic channels, its relation to a certain Rayleigh quotient, and a series of properties from the theory of semidefinite programming and quadratic forms. Unfortunately, these arguments do not seem applicable to our setup, which covers distributed precoding and random channels. Specifically, equations similar to \cite[Eq. (20)]{yu2007transmitter}, \cite[Eq. (25)]{yu2007transmitter}, and \cite[Eq. (29)]{yu2007transmitter} seem difficult to derive. To address this limitation, we follow a different path. We replace the above arguments by a variation of well-known uplink-downlink duality results under a sum power constraint, reviewed, e.g., in \cite{schubert2004solution,massivemimobook}.
\end{remark}  

\subsection{Joint precoding optimization over a dual uplink channel}
The desired connection between the downlink channel and its dual uplink channel is established by studying Problem~\eqref{prob:DL_QoS_augmented}, the solutions of which are optimal joint precoders solving Problem~\eqref{prob:DL_QoS}. The key idea lies in interpreting $\sum_{k=1}^K\E[\|\rvec{t}_k\|_{\vec{1}+\vec{\lambda}^\star}^2]$ as an unconventional weighted definition of the average sum transmit power. To keep the discussion general and, for instance, applicable to the algorithm given by Proposition~\ref{lem:proj_grad}, we consider arbitrary Lagrangian multipliers, i.e., we consider the following problem: $(\forall \vec{\sigma}\in \stdset{R}_{++}^L)$
\begin{equation}\label{prob:DL_QoS_augmented_general}
\begin{aligned}
\underset{\rvec{T} \in \set{T}} {\text{minimize}} \quad  & \sum_{k=1}^K \E[\|\rvec{t}_k\|_{\vec{\sigma}}^2]\\
\text{subject to} \quad & (\forall k \in \set{K})~\mathrm{SINR}^{\mathrm{DL}}_k(\rvec{T}) \geq \gamma_k.
\end{aligned}
\end{equation} 
Since the SINR constraints are feasible by assumption, following the same arguments as in the proof of Proposition~\ref{lem:proj_grad}, we observe that the above problem always admits a solution.

\begin{proposition}\label{prop:UL-DL}
For all $\vec{\sigma}\in \stdset{R}_{++}^L$, Problem~\eqref{prob:DL_QoS_augmented_general} and Problem~\eqref{prob:dualUL} have the same optimum.  Furthermore, a solution to Problem~\eqref{prob:DL_QoS_augmented_general} is given by
\begin{equation*}
(\forall k \in \set{K})~\rvec{t}_k^\star = \sqrt{q^\star_k}\rvec{v}_k^\star,
\end{equation*}
where $(\rvec{V}^\star,\vec{p}^\star)\in \set{T}\times\stdset{R}_{++}^K$ is a solution to Problem~\eqref{prob:dualUL}, and $\vec{q}^\star\eqdef (q_1^\star,\ldots,q^\star_K) \in \stdset{R}_{++}^K$ is given by
\begin{equation*}
\vec{q}^\star = (\vec{D}-\vec{B})^{-1}(\vec{D}-\vec{B}^\T)\vec{p}^\star,
\end{equation*}
where 
\begin{equation*}
\vec{B}\eqdef \begin{bmatrix}
\E[|\rvec{h}_1^\herm\rvec{v}^\star_1|^2] & \ldots & \E[|\rvec{h}_1^\herm\rvec{v}^\star_K|^2] \\
\vdots & \ddots & \vdots \\
\E[|\rvec{h}_K^\herm\rvec{v}^\star_1|^2] & \ldots & \E[|\rvec{h}_K^\herm\rvec{v}^\star_K|^2]
\end{bmatrix}, \quad \vec{D}\eqdef \mathrm{diag}(\vec{d}),
\end{equation*}
\begin{equation*}
\vec{d}\eqdef \left[(1+\gamma_1^{-1})|\E[\rvec{h}^\herm_1\rvec{v}^\star_1]|^2,\ldots,(1+\gamma_K^{-1})|\E[\rvec{h}^\herm_K\rvec{v}^\star_K]|^2\right].
\end{equation*}

\end{proposition}
\begin{proof}
The function $\E[\|\cdot\|_{\vec{\sigma}}^2]$ is a valid norm in $\set{T}_k$. Hence, we can rewrite $\inf_{\rvec{T} \in \set{T}_{\vec{\gamma}}} \sum_{k=1}^K\E[\|\rvec{t}_k\|_{\vec{\sigma}}^2]$ in a normalized form as the following optimization problem:
\begin{equation}
\begin{aligned}
\underset{\rvec{V} \in \set{T}, \vec{q} \in \stdset{R}_+^K} {\text{minimize}} \quad  &  \sum_{k=1}^Kq_k\\
\text{subject to} \quad & (\forall k \in \set{K})~\mathrm{SINR}^{\mathrm{DL}}_k(\rvec{V}\mathrm{diag}(\vec{q})^{\frac{1}{2}}) \geq \gamma_k,\\
& (\forall k \in \set{K})~\E[\|\rvec{v}_k\|_{\vec{\sigma}}^2] = 1,
\end{aligned}
\end{equation} 
where we used the change of variables $(\forall k \in \set{K})$ $\rvec{t}_k =: \sqrt{q_k}\rvec{v}_k$. The vector $\vec{q}$ can be interpreted as a downlink power control vector, by (unconventionally) measuring the power of each $\rvec{t}_k$ in terms of its norm $\E[\|\rvec{t}_k\|_{\vec{\sigma}}^2] = q_k$. For any choice of $\rvec{V}\in\set{T}$ with normalized columns, i.e., such that $(\forall k \in \set{K})~\E[\|\rvec{v}_k\|_{\vec{\sigma}}^2] = 1$, consider now the following downlink power control problem:
\begin{equation}\label{prob:DL_power_control}
\begin{aligned}
\underset{\vec{q} \in \stdset{R}_+^K} {\text{minimize}} \quad  &  \sum_{k=1}^Kq_k\\
\text{subject to} \quad & (\forall k \in \set{K})~ \mathrm{SINR}^{\mathrm{DL}}_k(\rvec{V}\mathrm{diag}(\vec{q})^{\frac{1}{2 }}) \geq \gamma_k.
\end{aligned}
\end{equation}
From known sum power duality arguments in the power control literature (reviewed, e.g., in \cite{schubert2004solution, massivemimobook}), it follows that Problem~\eqref{prob:DL_power_control} is feasible if and only if the following uplink power control problem is feasible, for the same choice of $\rvec{V}$:
\begin{equation}\label{prob:UL_power_control}
\begin{aligned}
\underset{\vec{p} \in \stdset{R}_+^K} {\text{minimize}} \quad  &  \sum_{k=1}^Kp_k\\
\text{subject to} \quad & (\forall k \in \set{K})~\mathrm{SINR}^{\mathrm{UL}}_k(\rvec{v}_k,\vec{p},\vec{\sigma}) \geq \gamma_k.
\end{aligned}
\end{equation}
When feasible, Problem~\eqref{prob:DL_power_control} and Problem~\eqref{prob:UL_power_control} are known to have unique and positive solutions meeting the SINR constraints with equality, and to attain the same optimum. The solutions are related by rearranging the constraints as full rank linear systems $(\vec{D}-\vec{B})\vec{q}^\star = \vec{1}$ and $(\vec{D}-\vec{B}^\T)\vec{p}^\star = \vec{1}$, respectively. When not feasible, we say that the two optima equal $\infty$. By taking the infimum of both optima over the set of $\rvec{V}\in\set{T}$ such that $(\forall k \in \set{K})~\E[\|\rvec{v}_k\|_{\vec{\sigma}}^2] = 1$, we obtain that Problem~\eqref{prob:dualUL} and Problem~\eqref{prob:DL_QoS_augmented_general} have the same optimum. Finally, the proof is completed by recalling that Problem~\eqref{prob:DL_QoS_augmented_general} always has a solution.
\end{proof}

The key message of the uplink-downlink duality principle in Proposition~\ref{prop:UL-DL} is that optimal joint precoders solving Problem~\eqref{prob:DL_QoS} are given by a scaled version of optimal joint combiners solving the dual uplink problem~\eqref{prob:dualUL} with the noise vector $\vec{\sigma}= \vec{1}+\vec{\lambda}^\star$, where $\vec{\lambda}^\star$ are  Lagrangian multipliers solving Problem~\eqref{prob:partial_dual}.

\subsection{Dual uplink power control with implicit optimal combining}\label{ssec:inner}
We now focus on the solution to the dual uplink problem~\eqref{prob:dualUL}. By exploiting the property that the uplink SINR constraints are only coupled via the power vector $\vec{p}$, we observe that the optimum to Problem~\eqref{prob:dualUL} is equivalently given by the optimum to the following power control problem:
\begin{equation}\label{prob:implicit_power_control}
\underset{\vec{p} \in \stdset{R}_+^K} {\text{minimize}}\: \sum_{k=1}^Kp_k  
\text{ subject to }  (\forall k \in \set{K})~u_k(\vec{p},\vec{\sigma})\geq \gamma_k,
\end{equation}
where the optimization of the combiners is left implicit in the definition of $(\forall k \in \set{K})(\forall \vec{p} \in \stdset{R}_+^{K})(\forall \vec{\sigma} \in \stdset{R}_{++}^L)$
\begin{equation}\label{eq:maxSINR}
u_k(\vec{p},\vec{\sigma}) \eqdef \sup_{\substack{\rvec{v}_k \in \set{T}_k \\ \E[\|\rvec{v}_k\|_{\vec{\sigma}}^2] \neq 0 } }\mathrm{SINR}^{\mathrm{UL}}_k(\rvec{v}_k,\vec{p},\vec{\sigma}).
\end{equation}
A solution $(\rvec{V}^\star,\vec{p}^\star)$ to Problem~\eqref{prob:dualUL} is related to a solution to Problem~\eqref{prob:implicit_power_control} in the sense that $\vec{p}^\star$ is also a solution to Problem~\eqref{prob:implicit_power_control}, and that the columns of $\rvec{V}^\star$ attain the suprema in \eqref{eq:maxSINR} for $\vec{p}=\vec{p}^\star$. Note that, without loss of generality and for mathematical convenience in later steps, we removed in \eqref{eq:maxSINR} the unit-norm constraint, because the value of $\E[\|\rvec{v}_k\|_{\vec{\sigma}}^2]$ does not change the uplink SINR, as long as it is non-zero.

The main implication of the above discussion is that optimal joint combiners (and hence, by Propostion~\ref{prop:UL-DL}, optimal joint precoders) solving Problem~\eqref{prob:dualUL} can be obtained by solving a set of disjoint uplink SINR maximization problems under per-AP information constraints, i.e., by evaluating all $u_k(\vec{p},\vec{\sigma})$ in \eqref{eq:maxSINR} for some coefficients $(\vec{p},\vec{\sigma})\in \stdset{R}_{++}^K\times \stdset{R}_{++}^L$. The next sections discuss challenges and solutions related to this crucial step, and they reveal a useful solution structure. However, before moving to the next section, we first present a numerical method for computing $\vec{p}^\star$, i.e., for solving Problem~\eqref{prob:implicit_power_control}, assuming that $u_k(\vec{p},\vec{\sigma})$ in \eqref{eq:maxSINR} can be indeed evaluated. Specifically, we apply the celebrated framework of interference calculus for power control \cite{yates1995power,martinbook}, and obtain:
\begin{proposition}\label{lem:fixed_point}
Fix $\vec{\sigma}\in \stdset{R}_{++}^L$. For every  $\vec{p}^{(1)} \in \stdset{R}_{++}^K$, the sequence $(\vec{p}^{(i)})_{i\in\stdset{N}}$ generated via the fixed-point iterations $
(\forall i \in \stdset{N})~\vec{p}^{(i+1)} = T_{\vec{\sigma}}(\vec{p}^{(i)})$, where
\begin{equation}\label{eq:int_map}
(\forall \vec{p}\in \stdset{R}_{++}^K)~T_{\vec{\sigma}}(\vec{p}) \eqdef \begin{bmatrix}
\frac{\gamma_1p_1}{u_1(\vec{p},\vec{\sigma})}, & \ldots, & \frac{\gamma_Kp_K}{u_K(\vec{p},\vec{\sigma})}
\end{bmatrix}^\T,
\end{equation}
converges geometrically in norm to the unique solution $\vec{p}^\star$ to Problem~\eqref{prob:implicit_power_control}.
\end{proposition}
\begin{proof} (Sketch) A simple contradiction argument proves that a solution $\vec{p}^\star$ must satisfy $(\forall k \in \set{K})~u_k(\vec{p}^\star,\vec{\sigma})= \gamma_k$. By trivially extending the arguments in \cite[Proposition~3]{miretti2022joint} for $\vec{\sigma}=\vec{1}$ to an arbitrary $\vec{\sigma}\in \stdset{R}_{++}^L$, it follows that this condition can be equivalently expressed as the fixed point equation $\vec{p}^\star = [\gamma_1f_{1,\vec{\sigma}}(\vec{p}^\star),\ldots,\gamma_Kf_{K,\vec{\sigma}}(\vec{p}^\star)]^\T$, where $(\forall k \in \set{K})~f_{k,\vec{\sigma}}: \stdset{R}_+^K\to \stdset{R}_{++}$ is a given positive concave function (and hence a \textit{standard interference function}, see, e.g., \cite[Definition~2]{miretti2022joint}) satisfying $(\forall k \in \set{K})(\forall \vec{p}\in \stdset{R}_+^K)~u_k(\vec{p},\vec{\sigma}) = p_k/f_{k,\vec{\sigma}}(\vec{p})$. For the above arguments to hold, an important requirement is the property $(\forall k \in \set{K})(\forall \vec{p}\in \stdset{R}_{++}^K)~u_k(\vec{p},\vec{\sigma}) > 0$. However, this property is guaranteed by recalling that Problem~\eqref{prob:DL_QoS_augmented_general} admits a solution, which implies that $\exists \rvec{T} \in \set{T}$ such that $(\forall k \in \set{K})~|\E[\rvec{h}_k^\herm\rvec{t}_k]|^2>0$, and hence the existence of some $\rvec{V} \in \set{T}$ such that all uplink SINRs are strictly positive for any $\vec{p}\in\stdset{R}_{++}^K$. The proof is concluded by invoking known properties of fixed points of standard interference mappings \cite{yates1995power} and positive concave mappings \cite{cavalcante2022positive}. In particular, the results in \cite{yates1995power} ensure that there exists a unique fixed point $\vec{p}^\star$, and that fixed-point iterations converge in norm to $\vec{p}^\star$. Furthermore, geometric convergence follows from \cite{cavalcante2022positive}.
\end{proof} 
The above proposition shows that optimal joint combiners solving Problem~\eqref{prob:dualUL} can be obtained via the iterative evaluation of $T_{\vec{\sigma}}(\vec{p})$ in \eqref{eq:int_map}, which in turn involves solving the aforementioned uplink SINR maximization problems in \eqref{eq:maxSINR}.

\section{Optimal joint precoding structure}
\label{sec:applications}
By exploiting the obtained uplink-downlink duality principle, we now derive the structure of an optimal solution to Problem~\eqref{prob:DL_QoS}. Specifically, we show that it suffices to consider properly scaled and regularized variations of the so-called \textit{team} MMSE precoders \cite{miretti2021team}, parametrized by a set of coefficients $(\vec{p},\vec{\sigma})\in \stdset{R}_{++}^K\times \stdset{R}_{++}^L$. 

\subsection{MMSE precoding under information constraints}\label{ssec:TMMSE}
As discussed in the previous sections, an optimal solution to Problem~\eqref{prob:DL_QoS} can be obtained by solving a set of uplink SINR maximization problems of the type in \eqref{eq:maxSINR}. Solving these problems appears quite challenging for the following reasons: (i) the non-convex fractional utility involving expectations in both the numerator and denominator, and (ii) the information constraints. However, the next proposition shows that (i) is not the main challenge, because we can consider an alternative and simpler convex utility, in the same spirit of the known relation between SINR and MMSE in deterministic channels \cite{shi2007duality}. Let us consider for all $k\in \set{K}$ the MSE between the data bearing signal $x_k\sim \CN(0,1)$ of UE~$k$ and its soft estimate $\hat{x}_k \eqdef \rvec{v}_k^\herm\rvec{y}$ obtained by processing the output of the virtual uplink channel $\rvec{y} \eqdef \sum_{k\in \set{K}}\sqrt{p_k}\rvec{h}_kx_k + \rvec{n}$ with noise $\rvec{n}\sim \CN(\vec{0},\vec{\Sigma})$, $\vec{\Sigma} \eqdef \mathrm{diag}(\sigma_1\vec{I}_N,\ldots,\sigma_L\vec{I}_N)$. The noise and all data bearing signals are mutually independent and independent of $(\rvec{H},\rvec{v}_1,\ldots,\rvec{v}_K)$. The MSE is defined as $(\forall k \in \set{K})(\forall \rvec{v}_k \in \set{T}_k)(\forall \vec{p}\in \stdset{R}_{++}^K)(\forall \vec{\sigma}\in \stdset{R}_{++}^K)$
	\begin{equation*}
		\begin{split}
			\mathrm{MSE}_k(\rvec{v}_k,\vec{p},\vec{\sigma}) \eqdef & \; \E[|x_k - \hat{x}_k|^2]\\
			=& \; \E\left[\|\vec{P}^{\frac{1}{2}}\rvec{H}^\herm\rvec{v}_k-\vec{e}_k\|_2^2\right] + \E\left[\|\rvec{v}_k\|_{\vec{\sigma}}^2\right],
		\end{split}
	\end{equation*} 
	where $\vec{P}\eqdef \mathrm{diag}(\vec{p})$, $\rvec{H}\eqdef[\rvec{h}_1,\ldots,\rvec{h}_K]$, and where the second equality can be verified via simple algebraic  manipulations. We then have:
\begin{proposition}
\label{prop:MSE}
For given $k\in\set{K}$, $\vec{p}\in\stdset{R}_{++}^K$, and $\vec{\sigma}\in \stdset{R}_{++}^L$, consider the optimization problem (recalling \eqref{eq:norm})
\begin{equation}\label{eq:MSE}
\underset{\rvec{v}_k \in \set{T}_k}{\emph{minimize}}~\mathrm{MSE}_k(\rvec{v}_k,\vec{p},\vec{\sigma}).
\end{equation}
Problem~\eqref{eq:MSE} has a unique solution $\rvec{v}_k^\star\in \set{T}_k$, and this solution satisfies
\begin{equation*}
u_k(\vec{p},\vec{\sigma}) = \mathrm{SINR}_k^{\mathrm{UL}}\left(\rvec{v}_k^\star,\vec{p},\vec{\sigma}\right) = \frac{1}{\mathrm{MSE}_k\left(\rvec{v}_k^\star,\vec{p},\vec{\sigma}\right)}-1.
\end{equation*}
\end{proposition}
\begin{proof}
The proof follows readily by extending the arguments in \cite[Proposition~4]{miretti2022joint} for $\vec{\sigma}=\vec{1}$ to an arbitrary $\vec{\sigma}\in \stdset{R}_{++}^L$. Additional details are reported in Appendix~\ref{app:MSE}.
\end{proof}
%Following standard arguments and simple algebraic manipulations, we can show (see, e.g., \cite{miretti2021team}) that the objective of Problem~\eqref{eq:MSE} corresponds to the MSE between the uplink signal of UE $k\in \set{K}$ and its soft estimate obtained by applying the joint combiner $\rvec{v}_k$ to the noisy received signal at the APs.
\begin{remark}
Together with Proposition~\ref{prop:UL-DL} and the discussion in Section~\ref{ssec:inner}, Proposition~\ref{prop:MSE} shows that optimal joint precoders solving Problem~\eqref{prob:DL_QoS} can be obtained as solutions to MMSE problems under information constraints, i.e., they are given by (properly scaled) solutions to  Problem~\eqref{eq:MSE}, for some parameters $(\vec{p},\vec{\sigma})\in \stdset{R}_{++}^K\times\stdset{R}_{++}^L$.
\end{remark}
In the particular case where all APs have complete knowledge of the channel $\rvec{H}$, and fully share UEs' data, a solution to Problem~\eqref{eq:MSE} is simply given by the following variation of the well-known regularized zero-forcing solution
\begin{equation}\label{eq:MMSE}
\rvec{V} = \left(\rvec{H}\vec{P}\rvec{H}^\herm + \vec{\Sigma}\right)^{-1}\rvec{H}\vec{P}^{\frac{1}{2}},
\end{equation}
where, similar to \cite{yu2007transmitter}, the typical regularization of the matrix inversion stage via a scaled identity matrix is replaced by a more general diagonal regularization matrix $\vec{\Sigma} = \mathrm{diag}(\sigma_1\vec{I}_N,\ldots,\sigma_L\vec{I}_N)$ parametrized by $\vec{\sigma}$ which takes into account the per-AP power constraints. 
\begin{remark}\label{rem:longterm}
A major difference between the solution in \cite{yu2007transmitter} and \eqref{eq:MMSE} is that the former considers short-term optimization of its parameters, i.e.,  for every channel realization. In contrast, our work considers long-term optimization of the parameters $(\vec{p},\vec{\sigma})\in \stdset{R}_{++}^K\times \stdset{R}_{++}^L$ based on channel statistics.
\end{remark}
For more general and nontrivial information constraints, the solution to Problem~\eqref{eq:MSE} can be interpreted as the best distributed approximation of regularized channel inversion. It can be obtained via a minor variation of the recently developed \emph{team} MMSE precoding method given in \cite{miretti2021team}. In the next sections we discuss this aspect in detail.

\subsection{Uplink channel estimation and CSI sharing}
\label{ssec:canonical_model}
To provide concrete examples of solutions to Problem~\eqref{eq:MSE}, we first slightly restrict the model in Section~\ref{sec:model}, while still covering most scenarios studied in the (cell-free) massive MIMO literature.  
For all AP $l\in \set{L}$ and UE $k\in\set{K}$, we let each sub-vector $\rvec{h}_{l,k}$ of $\rvec{h}_k^{\herm} =: [\rvec{h}_{1,k}^\herm, \ldots \rvec{h}_{L,k}^\herm]$ be independently distributed as $\rvec{h}_{l,k} \sim \CN\left(\vec{\mu}_{l,k}, \vec{K}_{l,k}\right)$ for some channel mean $\vec{\mu}_{l,k} \in \stdset{C}^N$ and covariance matrix $\vec{K}_{l,k}\in \stdset{C}^{N\times N}$. Independence can be easily motivated by assuming that UEs and APs are not colocated. As customary in the literature, we further assume that the APs (or, more generally, the processing units controlling the APs) perform pilot-based over-the-uplink MMSE channel estimation, based on the channel reciprocity property of time division duplex systems. Specifically, we assume each AP~$l \in \set{L}$ to acquire local estimates $\hat{\rvec{H}}_l\eqdef[\hat{\rvec{h}}_{l,1},\ldots,\hat{\rvec{h}}_{l,K}]$ of the local channel $\rvec{H}_l\eqdef[\rvec{h}_{l,1},\ldots,\rvec{h}_{l,K}]$ with error $\rvec{Z}_l \eqdef[\rvec{z}_{l,1},\ldots,\rvec{z}_{l,K}] \eqdef  \rvec{H}_l-\hat{\rvec{H}}_l$ independent from $\hat{\rvec{H}}_l$, and satisfying $(\forall l \in \set{L})(\forall k \in \set{K})$ $\rvec{z}_{l,k} \sim \CN(\vec{0},\vec{\Psi}_{l,k})$ for some error covariance $\vec{\Psi}_{l,k} \in \stdset{C}^{N\times N}$. Moreover, again motivated by the geographical separation of the APs, we assume that $(\hat{\rvec{H}}_l,\rvec{Z}_l)$ is independent from $(\hat{\rvec{H}}_j,\rvec{Z}_j)$ for all $(l,j)\in \set{L}^2$ such that $l\neq j$.

After this local channel estimation step, we assume that the APs may acquire additional information of the global channel matrix $\rvec{H}$ via some CSI sharing step. More specifically, we assume that each AP $l\in \set{L}$ must form its precoders based on some side information $S_l \eqdef (\hat{\rvec{H}}_l, \bar{S}_l)$, where $\bar{S}_l$ denotes additional channel information defined by the chosen CSI sharing pattern. Overall, we model this two steps channel acquisition scheme (local channel estimation followed by CSI sharing) by assuming the following Markov chain:
\begin{equation*}
(\forall l \in \set{L})(\forall j\in \set{L})\quad \rvec{H}_l - \hat{\rvec{H}}_l - S_l - S_j - \hat{\rvec{H}}_j - \rvec{H}_j.
\end{equation*}
The first and last part of this Markov chain correspond to the local channel estimation step, while the middle part corresponds to the subsequent CSI sharing step. Essentially, this Markov chain ensures that each AP $l\in \set{L}$ can only acquire a degraded version of the local estimate $\hat{\rvec{H}}_j$ of the local channel $\rvec{H}_j$ contained within the CSI $S_j$ of another AP $j\neq l$. In other words, AP $l\neq j$ cannot have a better knowledge of $\rvec{H}_j$ than AP $j$, and vice versa\footnote{Note that, while being natural for general time division duplex systems with reciprocity-based channel estimation, this assumption may not hold for frequency division duplex systems with nontrivial CSI feedback patterns. Since less attractive from both a practical and mathematical perspective, we do not consider the latter systems. However, we remark that these systems may also be covered by a variation of the team MMSE method in \cite{miretti2021team}.}. We map the above assumptions to information constraints $\set{T}_1,\ldots,\set{T}_K$ in Problem~\eqref{prob:DL_QoS} by letting the information subfield $\Sigma_l$ of each AP $l\in \set{L}$, defining the subspace $\set{H}_l^N$ in \eqref{eq:CSI}, be the sub-$\sigma$-algebra generated by its available CSI $S_l$ on $\Omega$.

\subsection{Centralized precoding with per-AP power constraints}
A similar expression to \eqref{eq:MMSE} covering user-centric network clustering and channel estimation errors can be easily obtained, provided that imperfect channel estimates are perfectly shared within each cluster of APs. This corresponds to a variation of the known centralized MMSE precoding solution in \cite{emil2020scalable,demir2021foundations}. 
\begin{proposition}
\label{prop:CMMSE}
For given $k\in\set{K}$, $\vec{p}\in\stdset{R}_{++}^K$, and $\vec{\sigma}\in \stdset{R}_{++}^L$, and under the model in Section~\ref{ssec:canonical_model} with $(\forall l\in \set{L})~S_l = (\hat{\rvec{H}}_1,\ldots,\hat{\rvec{H}}_L)$ (centralized CSI), the unique solution to Problem~\eqref{eq:MSE} is given by 
\begin{equation}\label{eq:CMMSE}
\rvec{v}_k = \left(\vec{C}_k\hat{\rvec{H}}\vec{P}\hat{\rvec{H}}^\herm\vec{C}_k+\vec{C}_k\vec{\Psi}\vec{C}_k+\vec{\Sigma}\right)^{-1}\vec{C}_k\hat{\rvec{H}}\vec{P}^\frac{1}{2}\vec{e}_k, 
\end{equation}
where $\hat{\rvec{H}}^\herm \eqdef[\hat{\rvec{H}}_1^\herm,\ldots,\hat{\rvec{H}}_L^\herm]$, $\vec{\Sigma}\eqdef \mathrm{diag}(\sigma_1\vec{I}_N,\ldots,\sigma_L\vec{I}_N)$, $\vec{\Psi} \eqdef \sum_{k=1}^Kp_k \mathrm{diag}(\vec{\Psi}_{1,k},\ldots,\vec{\Psi}_{L,k})$, and $\vec{C}_k \eqdef \mathrm{diag}(\vec{C}_{1,k},\ldots,\vec{C}_{L,k})$ is a block-diagonal matrix satisfying
\begin{equation*}
(\forall l \in\set{L})~\vec{C}_{l,k} = \begin{cases}\vec{I}_N & \text{if } l \in \set{L}_k,\\
\vec{0}_{N\times N} & \text{otherwise}.
\end{cases}
\end{equation*}
\end{proposition}
\begin{proof}(Sketch) We equivalently model user-centric network clustering by replacing $\rvec{v}_k$ with $\vec{C}_k\rvec{v}_k$ as in \cite{emil2020scalable}, instead of modifying the sets $\set{T}_k$ in \eqref{eq:CSI} as described in Section~\ref{ssec:info}. Then, since all APs share the same CSI $\hat{\rvec{H}}$, optimal precoders can be obtained as solutions to disjoint, unconstrained, and finite dimensional conditional MMSE problems, one for each CSI realization:
\begin{equation}\label{eq:condMMSE}
\rvec{v}_k = \arg\min_{\vec{v}_k \in \stdset{C}^{LN}}\E\left[\|\vec{P}^{\frac{1}{2}}\rvec{H}^\herm\vec{C}_k\vec{v}_k-\vec{e}_k\|^2+\|\vec{C}_k\vec{v}_k\|_{\vec{\sigma}}^2 \middle|\hat{\rvec{H}}\right].
\end{equation}
The rest of the proof follows by evaluating the conditional expectations using the CSI error model, and by standard results on unconstrained minimization of quadratic forms.
\end{proof}

Although derived by assuming full CSI sharing, we observe that the computation of the optimal precoder \eqref{eq:CMMSE} for UE $k\in \set{K}$ only requires knowledge of $(\hat{\rvec{H}}_l)_{l\in \set{L}_k}$, i.e., only the channel estimates of the APs serving UE $k$. Furthermore, similarly to the discussion in \cite{demir2021foundations}, the computation of the inverse in \eqref{eq:CMMSE} only involves the inversion of a submatrix of size $N|\set{L}_k|$.

\begin{remark}
The special case of \eqref{eq:CMMSE} with $\vec{\sigma}=\vec{1}$ was already proposed in \cite{emil2020scalable,demir2021foundations} as a good heuristic under a sum power constraint. In particular, it was motivated by first observing that \eqref{eq:CMMSE} maximizes a coherent uplink rate bound, different than \eqref{eq:UatF_bound}, and then by invoking uplink-downlink duality between \eqref{eq:UatF_bound} and \eqref{eq:hardening_bound} under a sum power constraint. Although \cite{emil2020scalable,demir2021foundations} observed that \eqref{eq:CMMSE} also solves Problem~\eqref{eq:condMMSE}, the connection with the maximization of \eqref{eq:UatF_bound} given by Propostion~\ref{prop:MSE}, and hence the formal optimality of \eqref{eq:CMMSE} in terms of downlink rates in \eqref{eq:hardening_bound}, was not reported. 
\end{remark}
Following similar arguments as in Remark~\ref{rem:longterm}, we further notice the following difference between \eqref{eq:CMMSE} and  similar expressions obtained using the technique in \cite{yu2007transmitter}, reported, for instance, in \cite{emil2013optimal}. Although both \eqref{eq:CMMSE} and these expressions involve a form of regularized channel inversion based on instantaneous CSI, the regularization parameters $(\vec{p},\vec{\sigma})$ in \eqref{eq:CMMSE} are selected based on channel statistics, instead of instantaneous CSI as in \cite{yu2007transmitter,emil2013optimal}. Despite some potential performance loss, this feature may simplify practical implementation, since it avoids performing complex parameter optimization (essentially, power control) for each channel realization.

\subsection{Distributed precoding with per-AP power constraints} 
When the APs have different CSI, it is not possible to decompose the problem into disjoint conditional MMSE problems as in the proof of \eqref{eq:CMMSE}, and more advanced methods must be used \cite{miretti2021team}. In this section we illustrate the extension of some key results in \cite{miretti2021team} to the case of per-AP power constraints. We start with the following general result:
\begin{proposition}
\label{prop:TMMSE}
For given $k\in\set{K}$, $\vec{p}\in\stdset{R}_{++}^K$, and $\vec{\sigma}\in \stdset{R}_{++}^L$, and under the model in Section~\ref{ssec:canonical_model}, the unique solution to Problem~\eqref{eq:MSE} is also the unique $\rvec{v}_k \in \set{T}_k$ satisfying $(\forall l \in \mathcal{L}_k)$
\begin{equation}\label{eq:TMMSE}
\rvec{v}_{l,k} =  \rmat{V}_l\left(\vec{e}_k-\sum_{j \in \set{L}_k\backslash  \{l\}} \vec{P}^{\frac{1}{2}}\E\left[\hat{\rmat{H}}_{j}^\herm\rvec{v}_{j,k}\Big|S_l\right] \right) \quad \mathrm{a.s.}, 
\end{equation} 
where $\rmat{V}_l:=\left(\hat{\rmat{H}}_{l}\vec{P}\hat{\rmat{H}}_{l}^\herm+\sum_{k\in\set{K}}p_k\vec{\Psi}_{l,k} + \sigma_l\vec{I}_N\right)^{-1}\hat{\rmat{H}}_{l}\vec{P}^{\frac{1}{2}}$. 
\end{proposition}
\begin{proof}
The proof follows readily by replacing $\vec{\sigma}=\vec{1}$ with an arbitrary $\vec{\sigma}\in \stdset{R}_{++}^L$ in the proofs of \cite[Theorem~3]{miretti2021team}, \cite[Lemma~2]{miretti2021team}, and \cite[Lemma~1]{miretti2021team2}. The proof is based on mapping Problem~\eqref{eq:MSE} to a slight generalization of the known class of so-called \textit{quadratic team decision} problems \cite{radner1962team,yukselbook}. Informally, \eqref{eq:TMMSE} is obtained by minimizing the objective in \eqref{eq:MSE} with respect to a subvector $\rvec{v}_{l,k}$, by conditioning on $S_l$, and by fixing the other subvectors $\rvec{v}_{j,k}$ for $j\neq l$. This readily gives a set of necessary optimality conditions, reminiscent of the game theoretical notion of Nash equilibium. The key step of the proof shows that these conditions are also sufficient.
\end{proof}
Proposition~\ref{prop:TMMSE} shows that optimal distributed precoding is composed by a local MMSE precoding stage, followed by a corrective stage taking into account the possibly unknown effect of the other APs based on the available CSI. From an optimization point of view,  Proposition~\ref{prop:TMMSE} states that the solution to \eqref{eq:MSE} is the unique solution to an infinite dimensional linear feasibility problem, that can be solved in closed form for many nontrivial yet relevant setups, or via efficient approximation techniques when the expectations in \eqref{eq:TMMSE} cannot be easily evaluated. We note that the difference between Proposition~\ref{prop:TMMSE} and the related results given in \cite{miretti2021team,miretti2021team2} for a sum power constraint is the introduction of a more flexible regularization of the local MMSE stage, which now involves the selection of a potentially different parameter $\sigma_l$ for each AP $l\in \set{L}$.

\begin{algorithm}[t]
\caption{for solving Problem~\eqref{prob:DL_QoS}}\label{algo}
\begin{algorithmic}[1]
\Require $\vec{p}\in \stdset{R}_{++}^K$, $\vec{\lambda}\in \stdset{R}_{+}^L$, $\alpha \in \stdset{R}_{++}$
\State $i\gets 1$
\Repeat
\State $\triangleright$ \textit{uplink joint power control and combining design}
\State $\vec{\sigma} \gets \vec{1}+\vec{\lambda}$
\Repeat	
\For{$k\in\set{K}$}
\State $\rvec{v}_k \gets \arg\min_{\rvec{v}_k\in\set{V}_k}\mathrm{MSE}_k(\rvec{v}_k,\vec{p},\vec{\sigma})$ 
\EndFor
\State $\vec{p} \gets \begin{bmatrix}
			\frac{\gamma_1p_1}{\mathrm{SINR}^{\mathrm{UL}}_1(\rvec{v}_1,\vec{p},\vec{\sigma})}, & \ldots, & \frac{\gamma_Kp_K}{\mathrm{SINR}^{\mathrm{UL}}_K(\rvec{v}_K,\vec{p},\vec{\sigma})}
		\end{bmatrix}^\T$ 
\Until{no significant progress is observed.}
\State $\rvec{V} \gets \begin{bmatrix}
\frac{\rvec{v}_1}{\sqrt{\E[\|\rvec{v}_1\|_{\vec{\sigma}}^2]}}, &\ldots, & \frac{\rvec{v}_K}{\sqrt{\E[\|\rvec{v}_K\|_{\vec{\sigma}}^2]}}
\end{bmatrix}$
\State $\triangleright$ \textit{uplink to downlink conversion}
\State $
\vec{B} \gets \begin{bmatrix} b_{j,k} 
\end{bmatrix}_{j\in \set{K},k\in \set{K}}$, where $b_{j,k} = \E[|\rvec{h}_j^\herm\rvec{v}_k|^2]$ 
\State $\vec{D} \gets \mathrm{diag}(d_k)_{k\in \set{K}}$, where $d_k = (1+\gamma_k^{-1})|\E[\rvec{h}^\herm_k\rvec{v}_k]|^2$
\State $\vec{q} \gets (\vec{D}-\vec{B})^{-1}(\vec{D}-\vec{B}^\T)\vec{p}$ 
\State $\rvec{T} \gets [q_1\rvec{v}_1,\ldots,q_K\rvec{v}_K]$
\State $\triangleright$ \textit{Lagrangian multipliers update}
\State $\vec{g} \gets [g_1,\ldots,g_L]^\T$ where $g_l = \sum_{k=1}^K\E[\|\rvec{t}_{l,k}\|^2]- P_l$
\State $\vec{\lambda} \gets \max\left(\vec{\lambda} + \frac{\alpha}{\sqrt{i}}\frac{\vec{g}}{\|\vec{g}\|},\vec{0}\right)$ element-wise
\State $i \gets i + 1$ 
\Until{no significant progress is observed.}
\end{algorithmic}
\begin{center}

\begin{tabular}{ |c|c| } 
 \hline
\textbf{Variant - line 7} & $\arg\min_{\rvec{v}_k\in\set{V}_k}\mathrm{MSE}_k(\rvec{v}_k,\vec{p},\vec{\sigma})$ \\ 
\hline
Centralized CSI & \eqref{eq:CMMSE} \\ 
Local CSI & \eqref{eq:LTMMSE} \\
Sequential CSI & solution to \eqref{eq:TMMSE} similar to \\
& \cite[Eq. (16)]{miretti2021team}, \cite[Eq.~(18)]{miretti2021team2} \\
\hline
\end{tabular}
\end{center}
\end{algorithm}

As for the results in \cite{miretti2021team,miretti2021team2}, we remark that Proposition~\ref{prop:TMMSE} applies to fairly general CSI sharing patterns (see Remark~\ref{rem:CSI}). Due to space limitations, in the remainder of this study we focus on the relatively simple case of local precoding with user-centric network clustering, and leave the study of more complex setups for future work. However, we point out that the same approach can be readily applied to extend the results on sequential precoding and undirectional CSI sharing given by  \cite{miretti2021team,miretti2021team2} to the case of per-AP power constraints.
\begin{proposition}\label{prop:LTMMSE}
For given $k\in\set{K}$, $\vec{p}\in\stdset{R}_{++}^K$, and $\vec{\sigma}\in \stdset{R}_{++}^L$, and under the model in Section~\ref{ssec:canonical_model} with  $(\forall l\in \set{L})~S_l = \hat{\rvec{H}}_l$ (local CSI), the unique solution to Problem~\eqref{eq:MSE} is given by 
\begin{equation}\label{eq:LTMMSE}
(\forall l\in \set{L})~\rvec{v}_{l,k} = \rvec{V}_l\vec{c}_{l,k}, 
\end{equation}
where $\rvec{V}_l$ is a local MMSE stage as in Proposition~\ref{prop:TMMSE}, and  $\vec{c}_{l,k}\in \stdset{C}^K$ is a statistical precoding stage given by the unique solution to the linear system of equations
\begin{equation*}
\begin{cases}\vec{c}_{l,k} + \sum_{j \in \set{L}_k \backslash \{l\}}\vec{\Pi}_j \vec{c}_{j,k} = \vec{e}_k & \forall l \in \set{L}_k, \\
\vec{c}_{l,k} = \vec{0}_{K\times 1} & \text{otherwise,}
\end{cases}
\end{equation*}
where $(\forall l\in \set{L})~\vec{\Pi}_l \eqdef \E\left[\vec{P}^{\frac{1}{2}}\hat{\rmat{H}}_l^\herm\rmat{V}_l\right]$.
\end{proposition}
\begin{proof}
The proof follows by verifying that \eqref{eq:LTMMSE} satisfies the optimality conditions \eqref{eq:TMMSE}. The details are similar to the full data sharing and sum power constraint case in \cite[Theorem~4]{miretti2021team}, and are reported in Appendix~\ref{app:local}.
\end{proof}
The computation of the local MMSE stage in \eqref{eq:LTMMSE} requires local instantaneous CSI, local statistical CSI, and  the parameters $(\vec{p},\vec{\sigma})$ selected based on global statistical CSI as for \eqref{eq:CMMSE}. The computation of the statistical precoding stage, whose role is to optimally enhance the local MMSE precoding decisions by exploiting additional statistical features of the channel and CSI of the other APs, requires global statistical CSI. A possible implementation could be to let a central processing unit perform the computation of the statistical precoding stages and the selection of $(\vec{p},\vec{\sigma})$, and to let the APs perform the remaining computations locally, without any instantaneous CSI sharing. Interestingly, we observe that the matrix $\vec{\Pi}_l$ for $l\in \set{L}$ takes the form of a $K\times K$ covariance matrix that can be locally estimated at the $l$th AP using local CSI only, and then shared on a long-term basis for the computation of the statistical precoding stages. As observed in \cite{miretti2021team} for the simpler case of a sum power constraint, \eqref{eq:LTMMSE} may provide significant performance gain over the local MMSE stage alone (where, for each $l\in \set{L}$ and $k\in \set{K}$, $\vec{c}_{l,k}$ is replaced by $c_{l,k}\vec{e}_k$ for a single tunable large-scale fading coefficient $c_{l,k}\in \stdset{R}_+$ \cite{demir2021foundations}) for non-zero mean channel models and/or in the presence of pilot contamination.

\section{Numerical examples}
\label{sec:sim}
In this section we illustrate some applications of our results by simulating the performance of optimal joint precoding in a simple example of user-centric cell-free massive MIMO network under either centralized or local CSI. Note that due to the theoretical nature and space limitations of this paper, the numerical results that follow are for illustrative purposes only and should by no means be understood as an exhaustive evaluation of the performance of cell-free networks.

\subsection{Simulation setup}
\label{ssec:sim_setup}
\begin{figure}
\centering
\includegraphics[width=0.8\columnwidth]{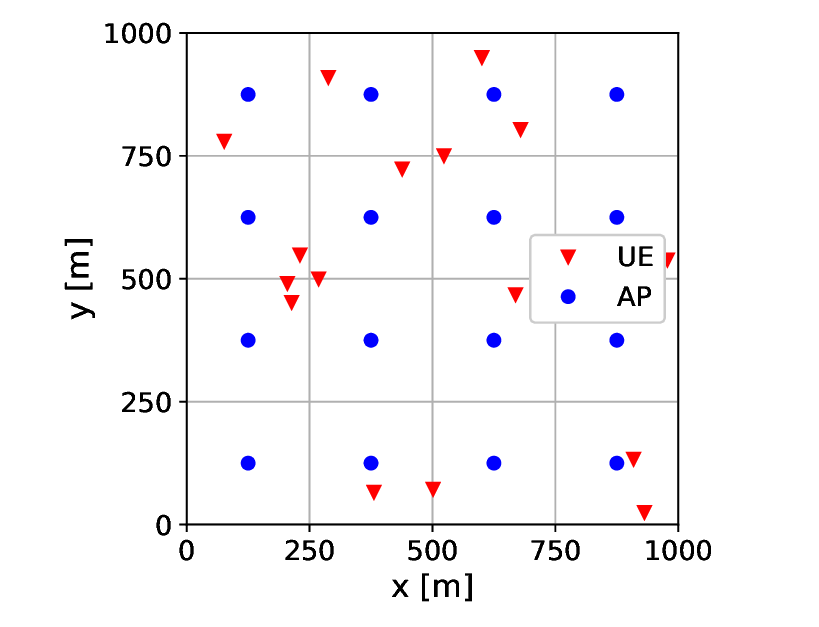}
\caption{Pictorial representation of the simulated setup: $K=16$ UEs uniformly distributed within a squared service area of size $1\times 1~\text{km}^2$, and $L=16$ regularly spaced APs with $N=4$ antennas each. Each UE is jointly served by a cluster of $Q=4$ APs offering the strongest channel gains.}
\label{fig:network}
\end{figure}
We consider the downlink of the network depicted in Figure~\ref{fig:network}, where $K=16$ UEs are independently and uniformly distributed within a squared service area of size $1\times 1~\text{km}^2$, and served by $L=16$ regularly spaced APs with $N=4$ antennas each. For all $(l,k)\in \set{L}\times \set{K}$, we assume for simplicity zero mean uncorrelated channels, i.e., $\vec{\mu}_{l,k} = \vec{0}$ and $\vec{K}_{l,k} = \kappa_{l,k}\vec{I}_N$, where $\kappa_{l,k}$ denotes the channel gain between AP~$l$ and UE~$k$. We adopt the following 3GPP-like channel gain model, suitable for an urban microcell scenario at $3.7$~GHz carrier frequency \cite[Table 7.4.1-1]{3GPP}
\begin{equation*}
\kappa_{l,k} = -35.3 \log_{10}\left(\Delta_{l,k}/1 \; \mathrm{m}\right) -34.5 + Z_{l,k} -P_{\mathrm{noise}} \quad \text{[dB]},
\end{equation*}
where $\Delta_{l,k}$ is the distance between AP $l$ and UE $k$ including a difference in height of $10$ m, and $Z_{l,k}\sim \mathcal{N}(0,\rho^2)$ [dB] are shadow fading terms with deviation $\rho = 7.82$. The shadow fading is correlated as $\E[Z_{l,k}Z_{j,i}]=\rho^22^{-\frac{\delta_{k,i}}{13 \text{ [m]}}}$ for all $l=j$ and zero otherwise, where $\delta_{k,i}$ is the distance between UE $k$ and UE $i$. The noise power is 
$P_{\mathrm{noise}} = -174 + 10 \log_{10}(B) + F$ [dBm],
where $B = 100$ MHz is the bandwidth, and $F = 7$ dB is the noise figure. The per-AP power constraints are set to $(\forall l \in\set{L})~P_l = 30$ dBm.

We assume that each UE~$k\in\set{K}$ is served by its $Q=4$ strongest APs only, i.e., by the subset of APs indexed by $\set{L}_k\subseteq \set{L}$, where each set $\set{L}_k$ is formed by ordering $\set{L}$ w.r.t. decreasing $\kappa_{l,k}$ and by keeping only the first $Q$ elements. Finally, we assume that the APs estimate the small-scale fading channel coefficients of the served UEs only. In particular, we consider the following simple model where  the channel coefficients are either perfectly known or completely unknown:
\begin{equation*}
(\forall k \in\set{K})~\hat{\rvec{h}}_{l,k} \eqdef \begin{cases}
\rvec{h}_{l,k} & \text{if } l \in \set{L}_k,\\
\E[\rvec{h}_{l,k}] & \text{otherwise}.
\end{cases}
\end{equation*}
The above model, sometimes referred to as the \textit{ideal partial CSI} model \cite{gottsch2022subspace}, neglects the impact of estimation noise and pilot contamination in the local channel acquisition step. We chose this model for simplicity and to keep our simulations focused on the impact of the chosen CSI sharing pattern, which is known to have a major effect on system performance \cite{demir2021foundations}. Nevertheless, it has been shown in \cite{gottsch2022subspace} that this model can lead to very accurate approximations of the performance of real-world systems using practical pilot-based over-the-uplink channel estimation schemes.

\subsection{Numerical methods for optimal joint precoding design}
The simulations in this section are produced by solving Problem~\eqref{prob:DL_QoS} using the numerical methods given by Proposition~\ref{lem:proj_grad} and Proposition~\ref{lem:fixed_point}, combined as nested loops. In particular, building on Proposition~\ref{prop:UL-DL}, the subgradients $\vec{g}(\vec{\lambda})$ in Proposition~\ref{lem:proj_grad} are computed by solving the downlink problem $\inf_{\set{T}_{\vec{\gamma}}}\sum_{k=1}^K\E[\|\rvec{t}_k\|_{\vec{1}+\vec{\lambda}}^2]$. This downlink problem is solved through its dual uplink problem \eqref{prob:dualUL} with $\vec{\sigma} = \vec{1}+\vec{\lambda}$, using the fixed-point algorithm in Proposition~\ref{lem:fixed_point}. Algorithm~\ref{algo} presents the considered routine under general information constraints. In our simulations, we specialize Algorithm~\ref{algo} to centralized and local precoding by performing the MSE minimization step in line 7 using the closed-form expressions given by \eqref{eq:CMMSE} in Proposition~\ref{prop:CMMSE} and by \eqref{eq:LTMMSE} in Proposition~\ref{prop:LTMMSE}, respectively. These expressions are evaluated under the clustering and local channel acquisition model presented in Section~\ref{ssec:sim_setup}, and by drawing a sample set of $100$ independent CSI realizations for approximating the expectations via empirical averages. An important practical aspect to underline is that, given that we are dealing with (infinite dimensional) functional optimization problems, storing the function $\rvec{v}_k$ in line 7 during the execution of the algorithm means storing its long-term parameters such as the statistical precoding stages $\vec{c}_{l,k}$ in \eqref{eq:LTMMSE}. 

Although our simulations are based on the simplified setup in Section~\ref{ssec:sim_setup}, we remark that the centralized and local precoding expressions \eqref{eq:CMMSE} and \eqref{eq:LTMMSE} can be readily applied to more general settings covering different clustering and local channel acquisition models. More precisely, this can be done by simply adapting the definitions of $\set{L}_k$, $\hat{\rvec{h}}_{l,k}$, and $\vec{\Psi}_{l,k}$ following the more general model in Section~\ref{ssec:canonical_model}. Furthermore, Algorithm~\ref{algo} can also be applied to more involved settings beyond centralized and local precoding, by solving the optimality conditions \eqref{eq:TMMSE} in Proposition~\ref{prop:TMMSE} as done, e.g., in \cite{miretti2021team, miretti2021team2} for the case of unidirectional CSI sharing. Moreover, we point out that Algorithm~\ref{algo} is by no means the only possible method for solving Problem~\eqref{prob:DL_QoS}. For example, the inner fixed-point iterations in line 9 could be replaced with faster methods for computing the fixed-point of the standard interference mapping $T_{\vec{\sigma}}$ \cite{renato2016elementary,boche2008}. In addition, the convergence of the outer projected subgradient iterations in line 19 can be improved by using known acceleration techniques \cite{polyak1969minimization,nesterov2003introductory}, or simply by choosing step size rules different than the current choice $(\forall i\in \stdset{N})~\alpha^{(i)} = \alpha/\sqrt{i}$. To avoid digressing from the main contribution of this study, i.e., the characterization of the optimal solution structure in Section~\ref{sec:applications}, we leave further analysis on the aforementioned extensions for future works. 

\subsection{Probability of feasibility}
\begin{figure}
\centering
\includegraphics[width=1\columnwidth]{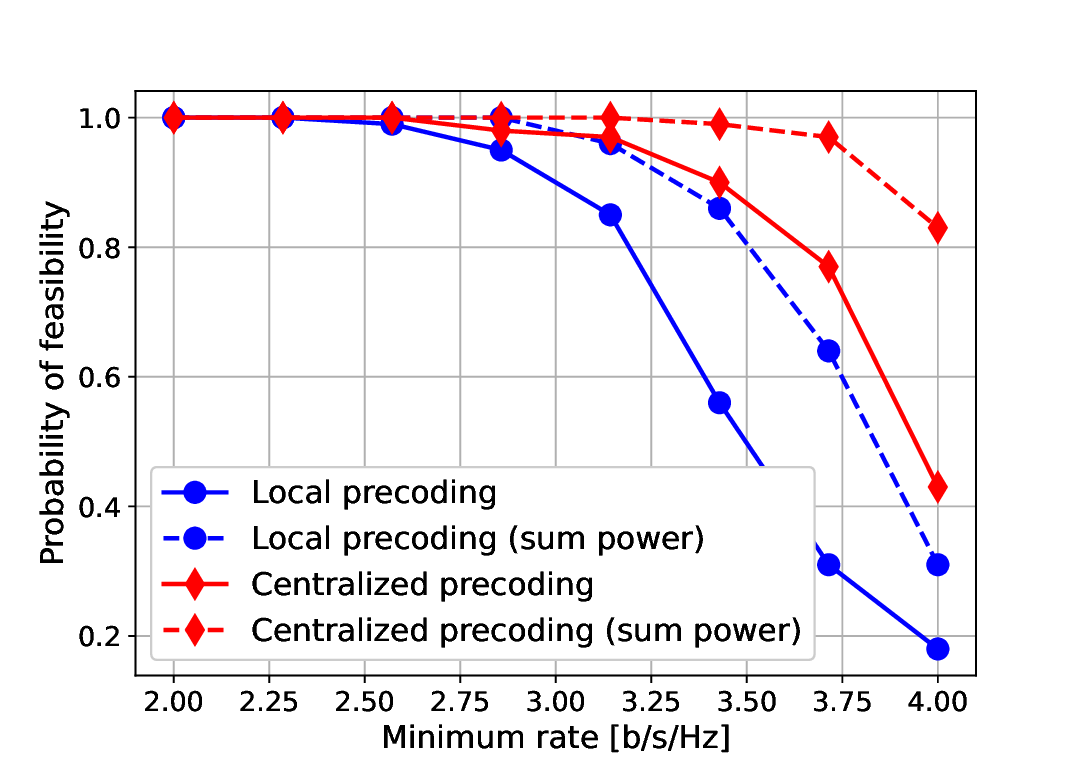}
\caption{Probability of feasibility of different minimum rate requirements, under different information constraints and a per-AP power constraint $P_l= 30$ dBm. Remarkably, our feasibility analysis covers optimal joint precoding under user-centric network clustering and either local CSI (local precoding) or centralized CSI (centralized precoding). The performance of the corresponding solutions under a sum power constraint $\sum_{l=1}^L P_l$ are also evaluated. As expected, due to the more restrictive information constraint, local precoding offers worse performance than centralized precoding. Similarly, a per-AP power constraints offers worse performance than a sum power constraint. Nevertheless, we observe that the local precoding with per-AP power constraint still robustly supports fairly high rates around $2.5$ b/s/Hz to all UEs.}
\label{fig:rate_feas}
\end{figure}
Figure~\ref{fig:rate_feas} plots the probability that the feasible set of Problem \eqref{prob:DL_QoS} is nonempty, by letting $(\forall k \in \set{K})~\gamma_k = \gamma$ for different choices of $\gamma$, corresponding to minimum rate requirements within $1$-$4$ b/s/Hz for all UEs.  We consider both centralized precoding as in Proposition~\ref{prop:CMMSE} and local precoding as in Proposition~\ref{prop:LTMMSE}. As a baseline, we also consider the  corresponding optimal precoders subject to a sum power constraint $P_{\text{sum}} = \sum_{l=1}^LP_l$. The probability of feasibility is approximated by solving $100$ instances of Problem~\eqref{prob:DL_QoS} for $100$ independent random UE drops using Algorithm~\ref{algo}. Although developed for producing a solution to Problem~\eqref{prob:DL_QoS} under the assumption of strict feasibility, we use the same algorithm for testing feasibility by monitoring its evolution. Note that, for the considered setup, the event of having a feasible set with empty interior has zero probability mass, hence we test feasibility by detecting the events of strict feasibility and infeasibility only. Additional details on how the algorithm performs this test are given below.
\begin{figure}[!t]
\centering
\subfloat[]{\includegraphics[width=3.5in]{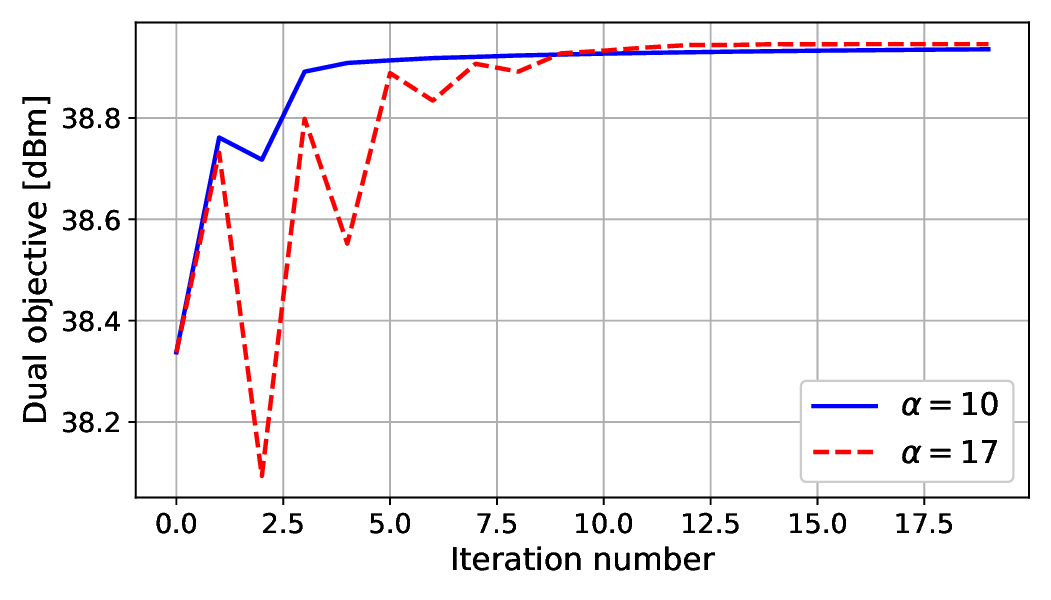}}
\hfil
\subfloat[]{\includegraphics[width=3.5in]{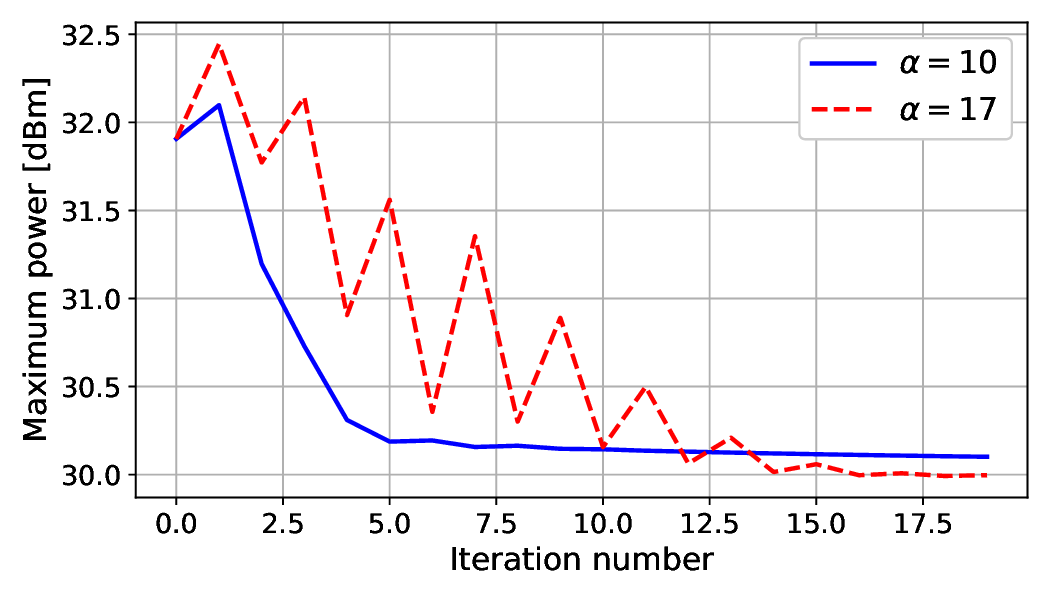}}
\caption{Example of convergence behavior of the proposed algorithm for different step size constants $\alpha$. We consider an arbitrary UE drop, local precoding, and a minimum rate requirement of $3.5$ b/s/Hz. For each iteration $i\in \stdset{N}$ of the outer loop, we plot (a) the dual objective $\tilde{d}(\vec{\lambda}^{(i)})$, and (b) the maximum transmit power over all APs $\max_{l\in\set{L}}g_l(\vec{\lambda}^{(i)})+P_l$. The non-monotonic convergence is a common feature of projected subgradient methods. We observe that, despite a seemingly slower convergence in the first iterations, the more aggressive step size choice $\alpha = 17$ produces a feasible solution satisfying the per-AP power constraint $(\forall l\in\set{L})~P_l = 30$ dBm after $20$ iterations only. This is enough to declare feasibility, since, for each outer iteration, the inner loops ensure that the SINR constraints are always satisfied.}
\label{fig:convergence}
\end{figure}

We first test feasibility under a sum power constraint, i.e., we test if $\inf_{\set{T}_{\vec{\gamma}}}\sum_{k=1}^K\E[\|\rvec{t}_k\|^2]\leq P_{\text{sum}}$ holds. To this end, we initialize the outer loop with $\vec{\lambda}^{(1)} = \vec{0}$, and the inner loop with $(\forall k \in \set{K})~p^{(1)}_k = \gamma_k/\sum_{l\in\set{L}}\kappa_{l,k}$. This choice ensures that $\vec{p}^{(1)}\leq T_{\vec{1}}(\vec{p}^{(1)})$ holds, and hence, by known properties of the considered fixed-point algorithm \cite[Fact~4]{renato2016elementary}, $(\sum_{k\in\set{K}}p_k^{(i)})_{i\in\stdset{N}}$ is a monotonically increasing sequence  converging to $\inf_{\set{T}_{\vec{\gamma}}}\sum_{k=1}^K\E[\|\rvec{t}_k\|^2]$ if $\set{T}_{\vec{\gamma}}\neq \emptyset$, or diverging to $+\infty$ if $\set{T}_{\vec{\gamma}}=\emptyset$ (unfeasible SINR requirements). The inner loop is terminated at some step $i\in \stdset{N}$ if no significant progress is observed, in which case we declare the problem feasible under a sum power constraint, or if the early stop condition $\sum_{k\in\set{K}}p_k^{(i)} > P_{\mathrm{sum}}$ is met, in which case we declare infeasibility under a sum power constraint.

If this initial feasibility test is passed, we continue with the outer loop, using the step size constant $\alpha=10$. Since the condition $\set{T}_{\vec{\gamma}}\neq \emptyset$ was already detected, the inner loops are guaranteed to compute the partial dual function $\tilde{d}(\vec{\lambda})=\inf_{\set{T}_{\vec{\gamma}}}\sum_{k=1}^K\E[\|\rvec{t}_k\|_{\vec{1}+\vec{\lambda}}^2]-\sum_{l=1}^L\lambda_lP_l$ and its subgradient $\vec{g}(\vec{\lambda})$ for all $\vec{\lambda}\in\stdset{R}_{+}^L$, up to some numerical tolerance. The outer loop is terminated at some step $i\in \stdset{N}$ if no significant progress is observed, or if the early stop condition $\tilde{d}(\vec{\lambda}^{(i)})>P_{\mathrm{sum}}$ is met. In the former case, and if the obtained solution is feasible up to some numerical tolerance, we declare feasibility. In all other cases, we declare infeasibility. In fact, if feasible, the obtained precoders are also a solution to Problem~\eqref{prob:DL_QoS}, i.e., a minimum sum power solution. In practice, when the interest is to test only feasibility, we terminate the algorithm as soon as a feasible solution is detected, i.e., if the additional early stop condition $\vec{g}(\vec{\lambda}^{(i)})\leq \vec{0}$ is met at some step $i\in\stdset{N}$. Mostly because of the heuristic step size constant $\alpha$, we remark that the algorithm might have slow convergence for some UE drops (see Figure~\ref{fig:convergence} for an illustration of the convergence behavior for a particular UE drop). Hence, we also introduce a maximum number of iterations and a maximum number of algorithm restarts with different constants $\alpha$. If a conclusion is not reached within these thresholds, we remove the corresponding UEs drop.

\section{Conclusion and future directions}
\label{sec:conclusion}
This study marks a step forward in the process of extending known analytic tools for deterministic channels and instantaneous rate expressions to fading channels and ergodic rate expressions. Such extensions are often advocated in the cellular and cell-free massive MIMO literature because they  allow for a more refined analysis of modern networks covering practical aspects such as imperfect CSI and system optimization based on long-term channel statistics. Just as coding over fading realizations is an efficient way to manage fading dips, system optimization based on the long-term perspective typically lead, in our opinion, to more efficient solutions. 

As first main result, this study advances the current understanding of distributed cell-free networks by showing that the recently introduced team MMSE precoding method provides joint precoders that are optimal not only under a sum power constraint, as stated in \cite{miretti2021team}, but also under a per-AP power constraint. For illustration purposes, this study derives the structure of optimal local precoding under a per-AP power constraint. Although not explicitly covered in this study, an optimal structure can be derived also for other examples of distributed precoding such as the ones based on sequential information sharing over radio stripes \cite{miretti2021team,miretti2021team2}. An interesting future direction is thus to revisit the available studies on performance evaluation of distributed cell-free downlink implementations, in light of the above results.

As second main result, this study provides an alternative tool to \cite{yu2007transmitter} for designing centralized cell-free networks subject to per-AP power constraints, by demonstrating optimality of centralized MMSE precoding with parameters tuned only once for many channel realizations, instead of for each channel realization as in the studies based on \cite{yu2007transmitter}. This is a consequence of considering the hardening bound as figure of merit. Although the hardening bound generally gives more pessimistic capacity estimates than alternative ergodic rate bounds based on coherent or semi-coherent decoding \cite{caire2018ergodic}, in our opinion this drawback is counterbalanced by a significantly enhanced tractability for precoding optimization.

A key aspect of our results is the compact parametrization of optimal joint precoding in terms of a set of coefficients $(\vec{p},\vec{\sigma})$, interpreted as virtual uplink transmit and noise powers. One limitation of this study is that the presented algorithm for tuning these parameters can be quite slow, and it is designed to solve SINR feasibility problems only (more specifically, Problem~\eqref{prob:DL_QoS}). Thus, another interesting future direction is the development of efficient algorithms, perhaps based on heuristics such as neural networks or other statistical learning tools, for tuning $(\vec{p},\vec{\sigma})$ in network utility maximization problems such as sum rate or minimum rate maximization.

Another subtle limitation of this work lies in the assumption of strict feasibility for the constraints of Problem~\eqref{prob:DL_QoS}, which is used for ensuring Slater's condition to hold in Proposition~\ref{prop:duality}. Using an information-theoretical perspective, this means that our results do not cover optimal joint precoding for rate tuples lying on the boundary of the achivable rate region, but only on its interior. However, from an engineering perspective, this limitation is practically irrelevant, since rate tuples arbitrarily close to the boundary are still covered. Nevertheless, resolving this limitation is an interesting yet involved future direction for theoretical analysis. 
 
\appendix
\subsection{Convex reformulation}\label{app:convex_reformulation}
Consider the reformulation of Problem~\eqref{prob:DL_QoS} obtained by replacing each (rearranged) SINR constraint in \eqref{eq:SINR_rearranged} with 
\begin{equation}\label{eq:SINR_reformulation}
(\forall k \in \set{K})~|c_k(\rvec{T})|-\nu_k\Re\left(b_k(\rvec{T})\right)\leq 0.
\end{equation}
More precisely, consider
\begin{equation}\label{prob:convex_reformulation}
\begin{aligned}
\underset{\rvec{T} \in \set{T}} {\text{minimize}} \quad &  \sum_{k=1}^K\E[\|\rvec{t}_{k}\|^2]\\
\text{subject to} \quad & (\forall k \in \set{K})~ |c_k(\rvec{T})|-\nu_k\Re\left(b_k(\rvec{T})\right)\leq 0\\
& (\forall l \in \set{L})~\sum_{k=1}^K\E[\|\rvec{t}_{l,k}\|^2]\leq P_l
\end{aligned}
\end{equation}
and its Lagrangian dual problem  
\begin{equation}\label{prob:dual_reformulation}
\underset{(\vec{\lambda},\vec{\mu})\in \stdset{R}_+^L\times \stdset{R}_+^K} {\text{maximize}}~d'(\vec{\lambda},\vec{\mu}),
\end{equation}
where we define the dual function $d'(\vec{\lambda},\vec{\mu}) \eqdef \inf_{\rvec{T} \in \set{T}} \Lambda'(\rvec{T},\vec{\lambda},\vec{\mu})$ and the Lagrangian $\Lambda'(\rvec{T},\vec{\lambda},\vec{\mu})\eqdef$
\begin{equation*}
\sum_{k=1}^K \E[\|\rvec{t}_k\|_{\vec{1}+\vec{\lambda}}^2] -\sum_{l=1}^L \lambda_l P_l + \sum_{k=1}^K\mu_k(|c_k(\rvec{T})|-\nu_k\Re(b_k(\rvec{T}))).
\end{equation*}
The main advantage of the above reformulation is  that it gives a convex optimization problem, as shown next.
\begin{lemma}\label{lem:convexity}
The objective and all constraints of Problem~\eqref{prob:convex_reformulation} are proper convex functions.\end{lemma}
\begin{proof}
Consider the norm $\| \cdot \|_{\star} \eqdef \sqrt{\langle \cdot,\cdot\rangle}$ on $\set{H}^{K+1}$ induced by the inner product $(\forall \rvec{x},\rvec{y} \in \set{H}^{K+1})~\langle \rvec{x},\rvec{y}\rangle := \Re\{\E[\rvec{y}^\herm\rvec{x}]\}$. We note that $|c_k(\rvec{T})|= \sqrt{\sum_{j=1}^K\E[|\rvec{h}_k^\herm\rvec{t}_j|^2]+1}$ is given by the composition of $\|\cdot\|_{\star}$ with an affine map $\set{T} \to \set{H}^{K+1}: \rvec{T}\mapsto [\rvec{h}_k^\herm\rvec{t}_1,\ldots,\rvec{h}_k^\herm\rvec{t}_K,1]^\T$, hence it is convex. Furthermore, since $\Re\left(b_k(\rvec{T})\right)=\Re\left(\E[\rvec{h}_k^\herm\rvec{t}_k]\right)$ is linear, convexity of the reformulated SINR constraints readily follows. We omit the proof for the convexity of the objective and power constraints, since it is trivial. Finally, repeated applications of Cauchy-Schwarz inequality prove that all the aforementioned functions are also proper functions.
\end{proof}

The next simple lemma can be used to relate Problem~\eqref{prob:DL_QoS} to Problem~\eqref{prob:convex_reformulation}, following a similar idea in \cite{wiesel2006linear,yu2007transmitter}.
\begin{lemma}\label{lem:phase_invariance}
Consider an arbitrary $\rvec{T}\in \set{T}$. Then, there exists $\rvec{T}' \in \set{T}$ such that $(\forall k\in\set{K})(\forall l \in\set{L})$
\begin{equation}
\begin{split}
|b_k(\rvec{T})| &= \Re\left(b_k(\rvec{T}')\right),\\
|c_k(\rvec{T})| &= |c_k(\rvec{T}')|, \\
\E[\|\rvec{t}_{l,k}\|^2] &= \E[\|\rvec{t}_{l,k}'\|^2].
\end{split}
\end{equation}
\end{lemma}
\begin{proof}
Observe that, $(\forall k\in\set{K})(\forall l \in\set{L})$, the terms $|b_k(\rvec{T})|$, $|c_k(\rvec{T})|$, and  $\E[\|\rvec{t}_{l,k}\|^2]$ are invariant to columnwise phase rotations of the argument, i.e., they do not vary if we replace $\rvec{T}$ with $[\rvec{t}_1 e^{j\theta_1},\ldots,\rvec{t}_K e^{j\theta_K}]$ for any $(\theta_1,\ldots,\theta_K)\in [0,2\pi]^K$. In particular, we can always pick $(\theta_1,\ldots,\theta_K)\in [0,2\pi]^K$ such that $|b_k(\rvec{T})| = \Re(b_k(\rvec{T}))$ holds.
\end{proof}
\begin{lemma}\label{lem:convex_optimum}
Let $p^\star$ and $r^\star$ be the optimum of Problem~\eqref{prob:DL_QoS} and Problem~\eqref{prob:convex_reformulation}, respectively. The two problems have the same optimum, i.e., $p^\star=r^\star$.
\end{lemma}
\begin{proof}
The simple property $(\forall x \in \stdset{C})~\Re(x)\leq |x|$ shows that $(\forall k \in \set{K})(\forall \rvec{T} \in \set{T})$
\begin{equation}\label{eq:constraint_inclusion}
|c_k(\rvec{T})| - \nu_k|b_k(\rvec{T})|\leq |c_k(\rvec{T})|-\nu_k\Re\left(b_k(\rvec{T})\right),
\end{equation}
from which the inequality $r^\star \geq p^\star$ readily follows (recall also \eqref{eq:SINR_rearranged}). Then, consider a minimizing sequence  for Problem~\eqref{prob:DL_QoS}, i.e., a (not necessarily convergent) sequence $(\rvec{T}^{(n)})_{n\in\stdset{N}}$ such that $(\forall n \in \stdset{N})$ $\rvec{T}^{(n)}$ satisfies all constraints of Problem~\eqref{prob:DL_QoS}, and $\lim_{n\to\infty} \sum_{k=1}^K\E[\|\rvec{t}^{(n)}_{k}\|^2] = p^\star$ \cite[Definition~1.8]{bauschke2011convex}. By Lemma~\ref{lem:phase_invariance}, we can define another sequence $(\rvec{T}^{(n)'})_{n\in\stdset{N}}$ such that $(\forall n \in \stdset{N})$ $\rvec{T}^{(n)'}$ satisfies all constraints of Problem~\eqref{prob:convex_reformulation},  attains the same objective $\sum_{k=1}^K\E[\|\rvec{t}^{(n)'}_{k}\|^2] = \sum_{k=1}^K\E[\|\rvec{t}^{(n)}_{k}\|^2] \geq r^\star$, and hence satisfies $p^\star =  \lim_{n\to\infty} \sum_{k=1}^K\E[\|\rvec{t}^{(n)'}_{k}\|^2] \geq r^\star$. Combining both inequalities $p^\star \leq r^\star$ and $p^\star \geq r^\star$ completes the proof.
\end{proof}

\begin{lemma}\label{lem:dual_equivalence}
The dual functions $d(\vec{\lambda},\vec{\mu})$ and $d'(\vec{\lambda},\vec{\mu})$ in Problem~\eqref{prob:DL_QoS_dual} and Problem~\eqref{prob:dual_reformulation}, respectively, satisfy $(\forall (\vec{\lambda},\vec{\mu})\in \stdset{R}_+^L\times \stdset{R}_+^K)$ $d(\vec{\lambda},\vec{\mu})= d'(\vec{\lambda},\vec{\mu})$.
\end{lemma}
\begin{proof}
Fix $(\vec{\lambda},\vec{\mu})\in \stdset{R}_+^L\times \stdset{R}_+^K$ and define the Lagrangian of Problem~\eqref{prob:DL_QoS} $\Lambda(\rvec{T},\vec{\lambda},\vec{\mu}) :=$
\begin{equation*}
\sum_{k=1}^K \E[\|\rvec{t}_k\|_{\vec{1}+\vec{\lambda}}^2] -\sum_{l=1}^L \lambda_l P_l + \sum_{k=1}^K\mu_k(|c_k(\rvec{T})|-\nu_k|b_k(\rvec{T})|),
\end{equation*}
which satisfies $d(\vec{\lambda},\vec{\mu})= \inf_{\rvec{T} \in \set{T}}\Lambda(\rvec{T},\vec{\lambda},\vec{\mu})$. The property  $(\forall x \in \stdset{C})~\Re(x)\leq |x|$ readily gives $\Lambda(\rvec{T},\vec{\lambda},\vec{\mu})\leq \Lambda'(\rvec{T},\vec{\lambda},\vec{\mu})$, and hence $d(\vec{\lambda},\vec{\mu})\leq d'(\vec{\lambda},\vec{\mu})$. Then, consider a minimizing sequence $(\rvec{T}^{(n)})_{n\in\stdset{N}}$ such that $(\forall n \in \stdset{N})$ $\rvec{T}^{(n)}\in\set{T}$ and $\lim_{n\to\infty} \Lambda(\rvec{T}^{(n)},\vec{\lambda},\vec{\mu}) = d(\vec{\lambda},\vec{\mu})$. By Lemma~\ref{lem:phase_invariance}, we can define another sequence $(\rvec{T}^{(n)'})_{n\in\stdset{N}}$ such that $(\forall n \in \stdset{N})$ $\rvec{T}^{(n)'} \in \set{T}$ and $\Lambda'(\rvec{T}^{(n)'},\vec{\lambda},\vec{\mu})=\Lambda(\rvec{T}^{(n)},\vec{\lambda},\vec{\mu})\geq d'(\vec{\lambda},\vec{\mu})$, hence satisfying $d(\vec{\lambda},\vec{\mu}) = \lim_{n\to\infty} \Lambda'(\rvec{T}^{(n)'},\vec{\lambda},\vec{\mu}) \geq d'(\vec{\lambda},\vec{\mu})$. Combining both inequalities $d(\vec{\lambda},\vec{\mu})\leq d'(\vec{\lambda},\vec{\mu})$ and $d(\vec{\lambda},\vec{\mu})\geq d'(\vec{\lambda},\vec{\mu})$ completes the proof.
\end{proof}

\begin{lemma}\label{lem:strong_duality}
Problem~\eqref{prob:convex_reformulation} admits a unique solution $\rvec{T}' \in \set{T}$. Furthermore, strong duality holds for Problem~\eqref{prob:convex_reformulation}, i.e., Problem~\eqref{prob:dual_reformulation} and Problem~\eqref{prob:convex_reformulation} have the same optimum, and there exist Lagrangian multipliers $(\vec{\lambda}',\vec{\mu}')$ solving Problem~\eqref{prob:dual_reformulation}.
\end{lemma}
\begin{proof}
From Lemma~\ref{lem:phase_invariance}, it immediately follows that since
Slater's condition is assumed to hold for Problem~\eqref{prob:DL_QoS}, then it holds also for Problem~\eqref{prob:convex_reformulation}. Then, strong duality and existence of Lagrangian multipliers follow by recalling Proposition~\ref{prop:duality} and Lemma~\ref{lem:convexity}. Existence of a unique solution $\rvec{T}'$ to Problem~\eqref{prob:convex_reformulation} follows by the Hilbert projection theorem (see, e.g., \cite{stark1998vector}), since the objective $\sum_{k=1}^K \E[\|\rvec{t}_k\|^2]$ is the norm induced by the inner product $(\forall \rvec{T},\rvec{V} \in \set{T})~\langle \rvec{T},\rvec{V}\rangle := \sum_{k=1}^K\Re\{\E[\rvec{v}_k^\herm\rvec{t}_k]\}$ on $\set{T}$, and the constraints define a nonempty closed convex subset of $\set{T}$  (closedness follows by continuity of the functions defining each constraint). In other words, in this Hilbert space, the solution to the problem is the projection of the zero vector onto the closed convex set defined by the constraints.
\end{proof}
The proof of Proposition~\ref{prop:strong_duality}
is readily given by combining Lemma~\ref{lem:convex_optimum}, Lemma~\ref{lem:dual_equivalence}, Lemma~\ref{lem:strong_duality}, and by noticing that the unique solution $\rvec{T}'$ to Problem~\eqref{prob:convex_reformulation} is also a solution to Problem~\eqref{prob:DL_QoS} (note: the converse does not hold in general).

\subsection{Recovering a primal solution from a dual solution}\label{app:primaldual}
Starting from the strong duality property in Proposition~\ref{lem:partial_dual}, we obtain that a primal-dual pair $(\rvec{T}^\star,\vec{\lambda}^\star)$ jointly solving Problem~\eqref{prob:DL_QoS} and Problem~\eqref{prob:partial_dual} must satisfy
\begin{equation*}
\begin{split}
p^\star &=  \inf_{\rvec{T} \in \set{T}_{\vec{\gamma}}} \sum_{k=1}^K \E[\|\rvec{t}_k\|_{\vec{1}+\vec{\lambda}^\star}^2] -\sum_{l=1}^L \lambda^\star_l P_l\\
&\leq \sum_{k=1}^K\E[\|\rvec{t}^\star_k\|_{\vec{1}+\vec{\lambda}^\star}^2] -\sum_{l=1}^L \lambda^\star_l P_l\\
&= \sum_{k=1}^K\E[\|\rvec{t}^\star_k\|^2]+\sum_{l=1}^L\lambda^\star_l\left(\sum_{k=1}^K\E[\|\rvec{t}^\star_{l,k}\|^2] - P_l\right)\\
&\leq \sum_{k=1}^K\E[\|\rvec{t}^\star_k\|^2] = p^\star,
\end{split}
\end{equation*}
where the first inequality follows by the definition of infimum, and the last inequality follows since $(\rvec{T}^\star,\vec{\lambda}^\star)$ satisfy the primal and dual constraints. The above chain of inequalities shows that $\rvec{T}^\star$ attains the infimum. If there was a unique $\rvec{T}\in \set{T}_{\vec{\gamma}}$ attaining the infimum, then it would also be the unique solution $\rvec{T}^\star$ to Problem~\eqref{prob:DL_QoS}. However, a similar property does not hold in general whenever the infimum is attained by multiple elements of $\set{T}_{\vec{\gamma}}$. In particular, the infimum may be attained not only by $\rvec{T}^\star$, but also by some other $\rvec{T}\in \set{T}_{\vec{\gamma}}$ violating the power constraints. Nevertheless, we now prove that this case can be here excluded, by using the next two lemmas. For convenience, we define the set $\set{T}'_{\vec{\gamma}}\subseteq\set{T}$ of precoders satisfying the reformulated SINR constraints \eqref{eq:SINR_reformulation}. 
\begin{lemma}\label{lem:inf_hilbert}
For every $\vec{\sigma}\in \stdset{R}_{++}^L$, there exists a unique $\rvec{T}'\in \set{T}'_{\vec{\gamma}}$ attaining $\inf_{\rvec{T} \in \set{T}'_{\vec{\gamma}}} \sum_{k=1}^K \E[\|\rvec{t}_k\|_{\vec{\sigma}}^2]$. 
\end{lemma}
\begin{proof}
The proof follows by the Hilbert projection theorem (see, e.g., \cite{stark1998vector}), since the objective $\sum_{k=1}^K \E[\|\rvec{t}_k\|_{\vec{\sigma}}^2]$ is the norm induced by the inner product $(\forall \rvec{T},\rvec{V} \in \set{T})~\langle \rvec{T},\rvec{V}\rangle := \sum_{k=1}^K\sum_{l=1}^L\sigma_l\Re\{\E[\rvec{v}_{l,k}^\herm\rvec{t}_{l,k}]\}$ on $\set{T}$, and the constraints define a nonempty closed convex subset of $\set{T}$. Specifically, in this Hilbert space, the infimum is attained by the projection of the zero vector onto the closed convex set $\set{T}'_{\vec{\gamma}}$.
\end{proof}
\begin{lemma}\label{lem:inf_relation}
For every $\vec{\sigma}\in \stdset{R}_{++}^L$, if $\rvec{T}$ attains $\inf_{\rvec{T} \in \set{T}_{\vec{\gamma}}} \sum_{k=1}^K \E[\|\rvec{t}_k\|_{\vec{\sigma}}^2]$, then $\exists \rvec{T}'$ attaining $\inf_{\rvec{T} \in \set{T}_{\vec{\gamma}}'} \sum_{k=1}^K \E[\|\rvec{t}_k\|_{\vec{\sigma}}^2]$ with the same power consumption. 
\end{lemma}
\begin{proof}
Following similar lines as for Lemma~\ref{lem:convex_optimum}, we can show that $ \inf_{\rvec{T} \in \set{T}_{\vec{\gamma}}} \sum_{k=1}^K \E[\|\rvec{t}_k\|_{\vec{\sigma}}^2] = \inf_{\rvec{T} \in \set{T}'_{\vec{\gamma}}} \sum_{k=1}^K \E[\|\rvec{t}_k\|_{\vec{\sigma}}^2]$ holds. The rest of the proof is a simple application of Lemma~\ref{lem:phase_invariance}.
\end{proof}
Informally, the rest of the proof exploits the power-preserving mapping between $\inf_{\rvec{T} \in \set{T}_{\vec{\gamma}}} \sum_{k=1}^K \E[\|\rvec{t}_k\|_{\vec{\sigma}}^2]$ and $\inf_{\rvec{T} \in \set{T}'_{\vec{\gamma}}} \sum_{k=1}^K \E[\|\rvec{t}_k\|_{\vec{\sigma}}^2]$ given by Lemma~\ref{lem:inf_relation}, and the fact that $\inf_{\rvec{T} \in \set{T}'_{\vec{\gamma}}} \sum_{k=1}^K \E[\|\rvec{t}_k\|_{\vec{\sigma}}^2]$ is a singleton by Lemma~\ref{lem:inf_hilbert}, to exclude that $\inf_{\rvec{T} \in \set{T}_{\vec{\gamma}}} \sum_{k=1}^K \E[\|\rvec{t}_k\|_{\vec{\sigma}}^2]$ is attained by points with different power consumption. More precisely, since the optimum $\rvec{T}^\star$ to the original Problem~\eqref{prob:DL_QoS} attains $\inf_{\rvec{T} \in \set{T}_{\vec{\gamma}}} \sum_{k=1}^K \E[\|\rvec{t}_k\|_{\vec{1}+\vec{\lambda}^\star}^2]$ and satisfies the power constraints, then, by Lemma~\ref{lem:inf_relation}, $\exists \rvec{T}'$ attaining $\inf_{\rvec{T} \in \set{T}_{\vec{\gamma}}'} \sum_{k=1}^K \E[\|\rvec{t}_k\|_{\vec{1}+\vec{\lambda}^\star}^2]$ and satisfying the power constraints. Since, by  Lemma~\ref{lem:inf_hilbert}, this $\rvec{T}'$ must be unique, there cannot be some $\rvec{T}\neq \rvec{T}^\star$ attaining $\inf_{\rvec{T} \in \set{T}_{\vec{\gamma}}} \sum_{k=1}^K \E[\|\rvec{t}_k\|_{\vec{1}+\vec{\lambda}^\star}^2]$ and violating the power constraints.

\subsection{Convergence of the projected subgradient method}
\label{app:proof_subgradient}
We apply \cite[Algorithm~3.2.8]{nesterov2003introductory} to the the minimization of $-\tilde d$ over $\vec{\lambda}\in \stdset{R}_{+}^L$. From known subgradient calculus rules, since $-\tilde{d}$ is the supremum of a family of affine functions indexed by $\rvec{T}\in\set{T}_{\vec{\gamma}}$, a subgradient at $\vec{\lambda}\in \stdset{R}^L$ is given by the gradient of any of these function attaining the supremum. This leads to the proposed algorithm. For all $ \vec{\lambda}\in \stdset{R}_+^L$, nonemptineess of $\arg\min_{\set{T}_{\vec{\gamma}}}\sum_{k=1}^K\E[\|\rvec{t}_k\|_{\vec{1}+\vec{\lambda}}^2]$ follows by combining: (i) $ \inf_{\rvec{T} \in \set{T}_{\vec{\gamma}}} \sum_{k=1}^K \E[\|\rvec{t}_k\|_{\vec{\sigma}}^2] = \inf_{\rvec{T} \in \set{T}'_{\vec{\gamma}}} \sum_{k=1}^K \E[\|\rvec{t}_k\|_{\vec{\sigma}}^2]$, similarly to Lemma~\ref{lem:convex_optimum}; (ii) Lemma~\ref{lem:inf_hilbert}; (iii) $\set{T}_{\vec{\gamma}}' \subseteq \set{T}_{\vec{\gamma}}$. Convergence of the best objective to the optimum $\tilde{d}(\vec{\lambda}^\star)$ for $n\to \infty$ follows from \cite[Lemma~3.2.1]{nesterov2003introductory} and the proof of \cite[Theorem~3.2.2]{nesterov2003introductory}, without using the Lipschitz continuity assumption. Furthermore, convergence of the corresponding argument follows since $\tilde{d}:\stdset{R}^L \to \stdset{R}$ is concave and hence continuous. 

\subsection{SINR maximization via MSE minimization}
\label{app:MSE}
Existence and uniqueness of the solution $\rvec{v}_k^\star$ follows from the Hilbert projection theorem (see, e.g., \cite{stark1998vector}) as in \cite[Lemma~4]{miretti2021team}. For the second statement, we observe that  
		\begin{align*}
			&\mathrm{MSE}_k(\rvec{v}^\star_k,\vec{p},\vec{\sigma})=\inf_{\substack{\rvec{v}_k\in\set{T}_k}} \mathrm{MSE}_k(\rvec{v}_k,\vec{p},\vec{\sigma})\\
			&\overset{(a)}{=} \inf_{\substack{\rvec{v}_k\in\set{T}_k\\ \E[\|\rvec{v}_k\|_{\vec{\sigma}}^2] \neq 0}} \inf_{\beta \in \stdset{C}} \mathrm{MSE}_k(\beta\rvec{v}_k,\vec{p},\vec{\sigma}) \\
			&\overset{(b)}{=} \inf_{\substack{\rvec{v}_k\in\set{T}_k\\ \E[\|\rvec{v}_k\|_{\vec{\sigma}}^2] \neq 0}} 1-\dfrac{p_k|\E[\rvec{h}_k^\herm\rvec{v}_k]|^2}{\sum_{j\in\set{K}}p_j\E[|\rvec{h}_j^\herm\rvec{v}_k|^2]+\E\left[\|\rvec{v}_k\|_{\vec{\sigma}}^2\right]} \\
			&=  \inf_{\substack{\rvec{v}_k\in\set{T}_k\\ \E[\|\rvec{v}_k\|_{\vec{\sigma}}^2] \neq 0}}\dfrac{1}{1+\mathrm{SINR}_k^{\mathrm{UL}}(\rvec{v}_k,\vec{p},\vec{\sigma})}= \dfrac{1}{1+u_k(\vec{p},\vec{\sigma})} 
		\end{align*}
		where $(a)$ follows from $(\forall \beta\in \stdset{C})(\forall \rvec{v}_k \in \mathcal{T}_k)$ $\beta \rvec{v}_k \in\mathcal{T}_k$, and $(b)$ follows from standard minimization of scalar quadratic forms. More specifically, $(b)$ follows by noticing that 
\begin{equation*}
\mathrm{MSE}_k(\beta\rvec{v}_k,\vec{p},\vec{\sigma}) = a|\beta|^2-2\Re(c^*\beta)+1,
\end{equation*}		
where $a \eqdef \sum_{j\in\set{K}}p_j\E[|\rvec{h}_j^\herm\rvec{v}_k|^2]+\E\left[\|\rvec{v}_k\|_{\vec{\sigma}}^2\right] \in \stdset{R}_{++}$ and $c \eqdef \E[\rvec{v}_k^\herm\rvec{h}_k]$, is a positive definite quadratic form in $\beta \in \stdset{C}$ with minimum $1-|c|^2/a$ attained at $\beta = c/a$. 

To conclude the proof we observe that, for the case $\rvec{v}_k^\star \neq \vec{0}$ (equivalently, $\E[\|\rvec{v}_k^\star\|_{\vec{\sigma}}^2] \neq 0$), the identity $u_k(\vec{p},\vec{\sigma}) = \mathrm{SINR}^{\mathrm{UL}}_k\left(\rvec{v}_k^\star,\vec{p},\vec{\sigma}\right)$ follows readily from the above chain of equalities. Furthermore, the case $\rvec{v}_k^\star=\vec{0}$ can be excluded due to the property $u_k(\vec{p},\vec{\sigma})>0$, already discussed in the proof of Proposition~\ref{lem:fixed_point}, since 
\begin{equation*}
u_k(\vec{p},\vec{\sigma})>0 \implies \mathrm{MSE}_k(\rvec{v}^\star_k,\vec{p},\vec{\sigma})<1 \implies \rvec{v}_k \neq \vec{0}.
\end{equation*}

\subsection{Optimal local precoding}
\label{app:local}
Substituting \eqref{eq:LTMMSE} into the  the optimality conditions \eqref{eq:TMMSE} for a given UE $k\in\set{K}$, we obtain $(\forall l \in \set{L}_k)$
\begin{align*}
\rvec{v}_{l,k} &=\rmat{V}_l\left(\vec{e}_k-\sum_{j \in \set{L}_k\backslash  \{l\}} \vec{P}^{\frac{1}{2}}\E\left[\hat{\rmat{H}}_{j}^\herm\rvec{V}_{j}\vec{c}_{j,k}\Big|\hat{\rvec{H}}_l\right] \right) \\
&= \rvec{V}_{l}\left(\vec{e}_{k} - \sum_{j \in \set{L}_k\backslash  \{l\}} \vec{\Pi}_j\vec{c}_{j,k} \right) = \rvec{V}_l \vec{c}_{l,k},
\end{align*}
where the second equality follows from the independence between $\hat{\rvec{H}}_l$ and $(\hat{\rvec{H}}_j,\rvec{V}_j)$ for $j\neq l$, and where the last equality follows since $(\forall l \in \set{L}_k)~\vec{c}_{l,k} + \sum_{j \in \set{L}_k\backslash  \{l\}} \vec{\Pi}_j\vec{c}_{j,k} = \vec{e}_k$ holds by construction. This last system of equations always admits a unique solution, as proven next.

We focus for simplicity on the case $\set{L}_k= \set{L}$, since the extension to arbitrary $\set{L}_k$ is immediate. We rewrite the system as $(\vec{M}+\vec{U}\vec{\Pi}) \vec{c}_k= \vec{U}\vec{e}_k$, where $\vec{c}_k:= \begin{bmatrix}
\vec{c}_{1,k}^\T & \ldots &\vec{c}_{L,k}^\T 
\end{bmatrix}^\T$, $\vec{\Pi}:= \begin{bmatrix}
\vec{\Pi}_1 & \ldots &\vec{\Pi}_L 
\end{bmatrix}$, $\vec{U}:= \begin{bmatrix}
\vec{I}_K & \ldots &\vec{I}_K 
\end{bmatrix}^\T $, $\vec{M}:=\mathrm{diag}(\vec{I}_K-\vec{\Pi}_1,\ldots,\vec{I}_K-\vec{\Pi}_L)$. We shall prove that $\vec{M}+\vec{U}\vec{\Pi}$ is invertible. By the matrix inversion lemma (see, e.g., \cite[Appendix~C]{boyd2004convex}), $\vec{M}+\vec{U}\vec{\Pi}$ is invertible if both $\vec{M}$ and $\vec{I}_K+\vec{\Pi}\vec{M}^{-1}\vec{U}$ are invertible. Furthermore, standard arguments show that $(\forall l \in \set{L})~\vec{0}\preceq \vec{P}^{\frac{1}{2}}\hat{\rvec{H}}_l^\herm\rvec{V}_l \preceq \vec{P}^{\frac{1}{2}}\hat{\rvec{H}}_l^\herm(\hat{\rmat{H}}_{l}\vec{P}\hat{\rmat{H}}_{l}^\herm + \sigma_l\vec{I}_N)^{-1}\hat{\rmat{H}}_{l}\vec{P}^{\frac{1}{2}} \prec \vec{I}_K$ holds almost surely, where we use $\preceq$ ($\prec$) to denote the partial ordering (strict partial ordering) with respect to the cone of Hermitian positive semidefine matrices as in \cite{boyd2004convex}. Therefore, $(\forall l \in \set{L})~\vec{0}\preceq \vec{\Pi}_l \prec \vec{I}_K$ holds. This in turn implies that $\vec{M}$ and $\vec{I}_K+\vec{\Pi}\vec{M}^{-1}\vec{U} = \vec{I}_K + \sum_{l\in \set{L}}\vec{\Pi}_l(\vec{I}-\vec{\Pi}_l)^{-1}$ are Hermitian positive definite, hence invertible.
(The eigenvalues of $\vec{\Pi}_l(\vec{I}-\vec{\Pi}_l)^{-1}$ take the form $\xi/(1-\xi)\geq 0$, where $0\leq\xi<1$ is an eigenvalue of $\vec{\Pi}_l$.)

\bibliographystyle{IEEEbib}
\bibliography{IEEEabrv,refs}

\end{document}